\documentclass[letterpaper, 10 pt, conference]{ieeeconf} 
\usepackage{cite}
\usepackage{amsmath,amssymb,amsfonts}
\usepackage{algorithmic}
\usepackage{graphicx}
\usepackage{algorithm,algorithmic}
\usepackage{mathtools}
\usepackage[hidelinks]{hyperref}
\usepackage{textcomp}
\usepackage{subcaption}
\usepackage{comment}
\usepackage{epstopdf}
\newtheorem{lemma}{Lemma}
\newtheorem{corollary}{Corollary}

\newtheorem{theorem}{Theorem}
\newtheorem{assmp}{Assumption}

\newtheorem{remark}{Remark}
\newtheorem{defn}{Definition}
\IEEEoverridecommandlockouts
\begin{document}

\title{On unified asymmetric barrier Lyapunov functions}
\author{Susmitha T Rayabagi, Shashi Ranjan Kumar, Debasattam Pal, and Dwaipayan Mukherjee
\thanks{Manuscript submitted to IEEE TAC. This work was supported in part by an ANRF-funded project bearing project code: CRG/2023/002280.}
\thanks{The authors Susmitha T Rayabagi, Debasattam Pal, and Dwaipayan Mukherjee are with the Electrical Engineering Department, Indian Institute of Technology Bombay, 400076, India (e-mail: srayabagi@iitb.ac.in, debasattam@ee.iitb.ac.in, dm@ee.iitb.ac.in).}
\thanks{Shashi Ranjan Kumar is with the Aerospace Engineering Department, Indian Institute of Technology Bombay, 400076, India (e-mail: srk@aero.iitb.ac.in).}
% \thanks{Third C. Author is with 
% the Electrical Engineering Department, University of Colorado, Boulder, CO 
% 80309 USA, on leave from the National Research Institute for Metals, 
% Tsukuba, Japan (e-mail: author@nrim.go.jp).}
}

\maketitle

\begin{abstract}
Barrier Lyapunov functions (BLFs) have been a popular choice when dealing with constrained control problems. In the current article, we present a \textit{unified} asymmetric barrier Lyapunov function that generalizes the existing logarithmic symmetric Lyapunov function. We show that the proposed function is smooth and does not require the discontinuous switching function that is ubiquitous in the asymmetric barrier Lyapunov functions existing in the literature. Based on the proposed barrier Lyapunov function, a control law, guaranteeing exponentially fast output tracking, is designed for a class of single-input single-output nonlinear systems with output constraints. We show that the proposed control law unifies the control design and structure for a system with either symmetric or asymmetric output constraints. Furthermore, we constructively show that the proposed function can be bounded from above and below by symmetric class $\mathcal{K}$ functions which help in establishing local exponential convergence together with the proposed control law. Lastly, we provide numerical examples to illustrate the performance of the proposed control and compare the results with the existing logarithmic BLFs. 
%The proposed function is a \textit{logarithmic} barrier function that accommodates both symmetric and asymmetric constraints and thus unifies the standard logarithmic barrier Lyapunov function for symmetric constraints. 
%Based on which we propose an exponentially stabilizing control scheme to achieve output tracking for a class of single-input single-output nonlinear systems with asymmetric output constraints.  Furthermore, the proposed control scheme unifies the control design and structure for a system with both symmetric or asymmetric output constraints, while ensuring exponential convergence. 
\end{abstract}


\section{Introduction}\label{sec:introduction}
{Constraints} arising in the form of saturation, safety limits and/or performance requirements etc, are an integral part of the control design process for any physical system. A common instance for saturation arises from the boundedness of force/acceleration provided by actuators. Furthermore, constraints related to safety and performance are of major importance in robotics and multi-agent swarm systems. Various techniques have been used to solve the constrained control problems, such as MPC \cite{Morari_MPC_1995,SolopertoAllgower_MPC_23}, Governors \cite{NONHOFFMuller_Ref_gov_2023}, control barrier functions \cite{RauscherHirche_CBFs_16,alanAmes_CBFs_23}, and artificial potential functions \cite{rimon1990exact_potential_1990, WolfBurdick_PotentialFuntions_08}. Along with these techniques, the use of barrier Lyapunov functions to solve constrained problems has garnered significant research interest. In \cite{Ngo_integratorBackstepping_1}, the authors introduced a logarithmic storage function with a barrier function's characteristics, in what would later be termed as symmetric barrier Lyapunov function (SBLF), to impose a hard bound on the velocity for a $4^{\text{th}}$ order integrator cascade in strict feedback form. This was motivated by the idealized model of the longitudinal dynamics of an aircraft. With the introduced SBLF, the authors further extended their work to consider bounds on multiple states of feedback linearizable systems in \cite{Ngo_integratorBackstepping_2}. Based on the SBLF in \cite{Ngo_integratorBackstepping_1}, an asymmetric barrier Lyapunov function (ABLF) was proposed in \cite{tee2009barrier} that accommodated different upper and lower bounds of the output constraint, i.e., asymmetric output constraints. The proposed ABLF is composed of two SBLFs, same as the one introduced in \cite{Ngo_integratorBackstepping_1}, stitched together using a discontinuous switching function. The SBLF and ABLF, introduced in \cite{Ngo_integratorBackstepping_1} and \cite{tee2009barrier}, respectively, have been widely used in various applications where symmetric and asymmetric constraints on states/outputs arise \cite{KPT_application_microactuators09,PurohitJain26,WenshanWenkaiWangSui_Omnidirectional26,ranjan2022generalized,AdelMohamadYassine_UAVs,ghoshBhasin26}. They were further used to design controllers for nonlinear SISO systems in strict feedback form, with time varying output constraints, in \cite{tee2011control} and for nonlinear switched systems in lower triangular form with time varying and constant output constraints in \cite{niu2013barrier}. Apart from the two aforesaid logarithmic functions, a few other varieties of BLFs have been studied. For instance, the authors in \cite{chenSun2020} and \cite{xu2012stateconstrained} introduced \textit{tangent}-type BLFs that incorporate symmetric constraints on output. In \cite{TANG2013449}, a \textit{tangent}-type ABLF that utilized the same discontinuous switching function as in \cite{tee2009barrier} was proposed to enforce asymmetric output constraints. Integral barrier Lyapunov functionals (iBLF) were proposed in \cite{tee2012control} that incorporated symmetric state constraints. From the above discussions, it is clear that the ABLFs studied in the literature so far rely on a discontinuous switching function to incorporate the asymmetric constraints. 

The discontinuous nature of the switching function gives rise to challenges such as ensuring smoothness of the function. Additional constraints are imposed to ensure smoothness of the ABLF and its successive derivatives \cite{tee2009barrier,TANG2013449}. These constraints are dependent on the order of the system and hinder the scalability of the designed control for higher order systems. Additionally, as the existing ABLFs do not naturally incorporate symmetric constraints, there is a strict distinction between symmetric and asymmetric BLFs. This necessitates the design of two different control structures based on the type of constraints. An exception to the use of the switching function is the \textit{universal barrier function} proposed in \cite{JinUniversalBarrier}, where a smooth rational polynomial function is employed to account for symmetric and asymmetric output constraints, as well as the unconstrained case. Aside from the smooth ABLF in \cite{JinUniversalBarrier}, alternative constructions of such smooth ABLFs remains largely unexplored. To this end, we propose a unified ABLF capable of addressing both symmetric and asymmetric constraints. Furthermore, we show that the proposed \textit{logarithmic} ABLF is a smooth function and generalizes the symmetric BLF of \cite{Ngo_integratorBackstepping_1} to the asymmetric case.

%These functionals allow the original state constraints to be mixed with the error terms, unlike most of the SBLFs discussed previously which used functions solely composed of errors. 
%Furthermore, almost all of the aforementioned results address asymptotic convergence, with the exception of [22], where exponential convergence was established for nonlinear systems with symmetric state constraints. The work in [17] also briefly addressed local exponential stability for asymmetric output constraints, however this was a special case of the time-varying output constraint. However, the proposed control scheme to achieve fixed-time convergence is considerably complex. The authors in [24] propose a function that is very similar in structure to that of the ABLF in [23]. However, it is a continuously differentiable function, unlike the smooth ABLF in [23], and was introduced to achieve finite time convergence for p-normal system with output constraints. Based of the proposed ABLF, we further propose a control scheme to achieve exponential convergence for nonlinear SISO systems in strict feedback form with output constraints.
Almost all of the aforementioned results address asymptotic convergence, with a few exceptions. For instance, the authors in \cite{tee2012control} established exponential convergence for nonlinear systems with symmetric state constraints. Further, local exponential stability for asymmetric output constraints was addressed briefly in \cite{tee2011control} as a special case of the time varying output constraints, and not as a primary focus of their study. The authors in \cite{JinUniversalBarrier} proposed a control scheme to achieve fixed-time convergence for multi-input multi-output nonlinear systems with the proposed \textit{universal barrier function}. However, the proposed control scheme, to achieve fixed-time convergence, is considerably involved. In \cite{chen2020unifiedBarrier_finitetime} finite time convergence for $p$-normal systems with output constraints was addressed. The ABLF considered was based on the \textit{universal barrier function} in \cite{JinUniversalBarrier}. However, the considered function is continuously differentiable, unlike the smooth ABLF in \cite{JinUniversalBarrier}. Observing that majority of the results in existing literature focus on asymptotic convergence, with a few exceptions as discussed, and further motivated by the local exponential stability guarantees in \cite{tee2011control} for asymmetric constraints, we propose a locally exponentially stabilizing controller that achieves output tracking while ensuring that the output constraints are not violated. The main contributions of this article are as follows:
\begin{enumerate}
    \item A unified \textit{logarithmic} ABLF is proposed. We show that the unified ABLF is smooth, strictly convex, positive definite and asymmetric. Furthermore, we also show that the proposed function reduces to the standard symmetric logarithmic BLF under symmetric constraints, thereby generalizing the existing symmetric BLF.
    %\item For the proposed ABLF, we show the existence of three important functions that will facilitate the design of exponentially stabilizing control  
    \item Based on the proposed ABLF, we propose a control law for nonlinear single-input single-output (SISO) systems in strict feedback form with output constraints. The proposed control law is unified in the sense that a single control law can handle both symmetric and asymmetric output constraints, and overcomes the challenges that come with the use of the switching function in \cite{tee2009barrier}.
    \item We constructively show that the proposed function can be bounded from above and below by symmetric class $\mathcal{K}$ functions. Along with these bounds we show the existence of an alternate upper bounding function, which generalizes a logarithmic inequality in \cite{TeeRen_inequa_2010}. These constructions form an important tool in establishing local exponential stability.
    \item We present two numerical examples with difference in system order to illustrate the scalability of the proposed scheme. Furthermore, we also comment on the difference in performance of the proposed control scheme with the one existing in \cite{tee2009barrier}, with symmetric and asymmetric constraints.
\end{enumerate}
\section{Preliminaries and Problem Formulation}\label{Sec:Problem_Formulation}
\subsection{System description}\label{subSec:Prob_form_sys_description}
Consider a nonlinear system of order $n$ in the strict feedback form given below:   
\small
\begin{equation}\label{Eqn:sys_dynamics}
\begin{split}
    & \dot{x}_{i}=f_{i}(x_1,\dots,x_{i})+g_{i}(x_1,\dots,x_{i})x_{i+1}, ~1\leqslant i\leqslant n-1\\
    & \dot{x}_{n}=f_{n}(x_1,\dots,x_{n})+g_{n}(x_1,\dots,x_{n})u,\\
    & y=x_1,
\end{split}
\end{equation}
\normalsize
where the functions $f_1,\dots,f_n$ and $g_1,\dots,g_n$ are smooth (refer \textit{Definition} \ref{Def:smooth_function} given below), $x_1,\dots,x_n$ are the states, and $u$ and $y$ are the input and output, respectively. Let ${x}\coloneqq [x_1,x_2,\dots ,x_n]^T~\in \mathbb{R}^n$ denote the state vector and ${x}_{2:n}\coloneqq [x_2,\dots,x_n]\in \mathbb{R}^{n-1}$ denote all the states except $x_1$. The constraint on output is given by: $|y(t)|<k_y$ for all $t\geqslant 0$, where $k_y>0$.

Note: The set of nonnegative real numbers is denoted by $\mathbb{R}_+$ and $\|\cdot\|_2$ denotes the Euclidean 2-norm in $\mathbb{R}^n$.

The aim is to design a control law $u(t)$ to achieve output tracking of a given desired trajectory $y_d(t)$ such that the output constraint is not violated. We make the following assumptions on the desired trajectory $y_d(t)$ and the system described in \eqref{Eqn:sys_dynamics}. 
\begin{assmp}\label{Assumption:bounds_on_yd}
For any $k_y\!>\!0$, there exist positive constants $Y_{lb},Y_{ub},Y_1,Y_2,\dots,Y_n$ that satisfy the condition $\text{max}(Y_{lb},Y_{ub})\!<\!k_y$, such that $y_d(t)$ and its successive time derivatives satisfy $-Y_{lb}\!<\!y_d\!<\!Y_{ub}$, $|y^{(i)}_d(t)|<Y_i,\forall t\geqslant0, \forall i\in \{1,2,\dots,n\}$, where $y^{(i)}_d(t)$ is the $i^{\text{th}}$ time derivative of $y_d$. 
\end{assmp}
\begin{assmp}\label{Assumption:bounds_on_gi}
    The functions, $g_1,g_2,\dots,g_n$ are known and $\exists g_0\in \mathbb{R}_+$ such that $0\!<\!g_0\!<\!|g_{i}(x_1,\dots,x_i)|$ $\forall i\in \{1,\dots,n\}$, for $|x_1|\!<\!k_y$. Furthermore, without loss of generality we assume that $g_i(x_1,\dots,x_i)$ are all positive for $|x|<k_y$.
\end{assmp}
% \begin{assmp}\label{Assumption:LIP_uncertain_system}
%     The functions $f_i(x_1,x_2,\dots,x_i)$ are uncertain, but they satisfy the linear-in-the-parameters condition
%     \begin{align}\label{Eqn:LIP_condition}
%         f_i(x_1,x_2,\dots,x_i)=\theta^T\psi(x_1,x_2,\dots,x_i),
%     \end{align}
%     where $\psi_1,\psi_2,\dots,\psi_n$ are smooth functions and $\theta\in \mathbb{R}^l$ is the uncertain parameters vector which satisfies $\|\theta\|\!<\!\theta_M,~\theta_M\!>\!0$.
% \end{assmp}
We now state some relevant definitions from \cite{krantz1999geometry} and \cite{khalil}.
\begin{defn}\label{Def:smooth_function}
A function $f \colon \mathbb{R}^n\to \mathbb{R}$ is said to be continuously differentiable of order $k$, or $\mathcal{C}^k$, if $$D^af \coloneqq\frac{\partial^{a_1}}{\partial x_1^{a_1}}\frac{\partial^{a_2}}{\partial x_2^{a_2}}\cdots\frac{\partial^{a_n}}{\partial x_n^{a_n}}f$$ exists and is continuous, for all points $(x_1, x_2, . . . , x_n)$ in $\mathbb{R}^n,$ and all nonnegative integers $a_1, a_2,\ldots, a_n$, satisfying $\sum_{i=1}^n a_i \le k$. A smooth, or $\mathcal{C}^{\infty}$, function $f \colon \mathbb{R}^{n}\to \mathbb{R}$ is one that is $\mathcal{C}^k$ for every positive integer $k$.
\end{defn}
\begin{defn}\label{Def:class_K_inifity}
A continuous function $\alpha\colon[0,a)\longrightarrow [0,\infty)$ is said to belong to class $\mathcal{K}$ if it is strictly increasing and $\alpha(0)=0$. It is said to belong to class $\mathcal{K}_{\infty}$ if $a=\infty$ and $\alpha(r)\longrightarrow\infty$ as $r\longrightarrow \infty$.
\end{defn}
\subsection{Preliminary results}
\subsubsection{Infinite logarithmic series expansion\texorpdfstring{\cite{spivak2006calculus}}{\cite{spivak2006calculus}}}
Given $\hat x\in (-1,1]$, the series expansion for $\text{log}(1+\hat x)$ is given by:
\begin{align}\label{Eqn:series_expand_log}
    \text{log}(1+\hat x)=\sum_{k=1}^{\infty}(-1)^{k+1}\frac{x^k}{k},
\end{align}
where log$(\bullet)$ denotes the natural logarithm of $(\bullet)$. From the ratio test \cite[Theorem 3, Ch 23]{spivak2006calculus} it can be shown that the series is absolutely convergent for $|\hat x|<1$.
For the series expansion of $\text{log}(1-\hat{x})$, substitute $-\hat x$ in \eqref{Eqn:series_expand_log} with $\hat{x}\in [-1,1)$.
\subsubsection{Cardan's formula for cubic polynomials\cite{Uspensky1948TheoryOfEquations}}\label{subSec:cardans_formula}
Consider the general cubic equation given by
\begin{align}\label{Eqn:general_cubic_polynomail}
    f(x)=x^3+\tilde ax^2+\tilde bx+\tilde c=0.
\end{align}
Substituting $y=x-(\tilde a/3)$ the given general cubic polynomial can be transformed to a depressed cubic polynomial
\begin{align}\label{Eqn:depressed_cubic_polynomial}
    y^3+py+q=0,~~p=\tilde b\!-\!\frac{\tilde a^2}{3},~q=\tilde c\!-\!\frac{\tilde b \tilde a}{3}+\frac{2\tilde a^3}{27}.
\end{align}
The above cubic polynomial has the following roots:
\begin{align}
    y_1&=\!\sqrt[3]{A}\!+\!\sqrt[3]{B},~~y_2\!=\!\omega\sqrt[3]{A}\!+\!\omega\sqrt[3]{B},~~y_3\!=\!\omega^2\sqrt[3]{A}\!+\!\omega^2\sqrt[3]{B},\notag
\end{align}
where $\omega=(-1+i\sqrt{3})/2$ is an imaginary cube root of unity, and $A$ and $B$ are given by:
\begin{align}\label{Eqn:A_B_def_roots}
    A\coloneqq -\frac{q}{2}+\sqrt{\!\frac{q^2}{4}\!+\!\frac{p^3}{27}},~B\coloneqq -\frac{q}{2}-\sqrt{\!\frac{q^2}{4}\!+\!\frac{p^3}{27}}.
\end{align}
Since $p,q\in\mathbb{R}$, the nature of the roots depend on the $\Delta\!\coloneqq \!4p^3\!+\!27q^2$. This is clear from the definition of $A$ and $B$ in \eqref{Eqn:A_B_def_roots}. When $\Delta\!>\!0$, the cubic polynomial in \eqref{Eqn:depressed_cubic_polynomial} will have one real root and two complex conjugate roots. The case when $\Delta\!\leqslant \!0$ is not germane to the current work and will not be considered. However, interested readers may refer to Ch. 5 of \cite{Uspensky1948TheoryOfEquations} for a detailed description of the steps for obtaining the roots of \eqref{Eqn:depressed_cubic_polynomial}.
\subsubsection{Stability with general forms of barrier functions\cite{tee2009barrier}}\label{subSec:Barrier_results}
% \begin{defn}[\cite{tee2009barrier}]
% A barrier Lyapunov function is a scalar function $V(x)$, defined with respect to the system $\dot{x} = f(x)$ on an open region, $D$, containing the origin, that is continuous, positive definite, has continuous first-order partial derivatives at every point of $D$, has the property $V (x) \to\infty$ as $x$ approaches the boundary of $D$, and satisfies $V (x(t)) \leqslant b~~ \forall t \geqslant 0$ along the solution of $\dot{x} = f (x)$ for $x(0) \in D$ and some positive constant $b$.
% \end{defn}
The following result formalizes the stability guarantees for the design of control for the system in a strict feedback form \eqref{Eqn:sys_dynamics}, using general forms of barrier functions, such that the output or state constraints are not violated. 
\begin{lemma}\label{lemma:lemma_1_KPT}
For arbitrary positive constants $k_{a_1}, k_{b_1}$, define the following two open subsets of $\mathbb{R}^{l+1}$: $\mathcal{Z}_1 \coloneqq \left\{z_1 \in \mathbb{R}: -k_{a_1} <z_1 <k_{b_1}\right\}\subset \mathbb{R}$ and
$\mathcal{N} \coloneqq \mathbb{R}^l \times \mathcal{Z}_1$. Consider the system
$$\dot \eta= h(t, \eta),$$ where $\eta \coloneqq [w,~~z_1]^T\in \mathcal{N}$, and $h \colon \mathbb{R}_+ \times \mathcal{N} \to \mathbb{R}^{l+1}$ is piecewise continuous in $t$ and locally Lipschitz in $z_1$, uniformly in $t$, on $\mathbb{R}_+ \times \mathcal{N}$. Suppose that there exist functions $U \colon \mathbb{R}^l \to \mathbb{R}_+$ and $V_1 \colon \mathcal{Z}_1 \to \mathbb{R}_+$, continuously differentiable and positive definite in their respective domains, such that $V_1(z_1)\to\infty~\text{as}~ z_1\to -k_{a_1}~\text{or}~z_1\to k_{b_1},$ and $\gamma_1(||w||)\le U(w)\le \gamma_2(||w||),$ where $\gamma_1$ and $\gamma_2$ are class $\mathcal{K}_{\infty}$ functions. Let $V(\eta) \coloneqq V_1(z_1) + U(w)$, and let $z_1(0) \in (-k_{a_1}, k_{b_1})$. If the inequality
\begin{align}
   \dot V = \dfrac{\partial V}{\partial \eta}h\le 0
\end{align}
holds over $\mathcal{N}$, then $z_1(t)$ remains in the open set $z_1 \in (-k_{a_1}, k_{b_1})$ $\forall\,\,t \,\in [0,\infty)$.
\end{lemma}
\subsubsection{Exponential stability\cite{blanchini2008set}}
Consider the system 
\begin{align}\label{Eqn:sys_exponential_stability_prelim}
    \dot\xi(t)=f(t,\xi), \xi\in \mathbb{R}^n,
\end{align}
where $f\colon [0,\infty)\times D\longrightarrow \mathbb{R}^n$ is piecewise continuous in $t$ and locally Lipschitz in $\xi$ on $[0,\infty)\times D$, and $D\subset \mathbb{R}^n$ is a domain containing $\xi=0$. 
\begin{theorem}\label{theorem:exponetial_stability}
Let $\xi=0$ be the equilibrium point for \eqref{Eqn:sys_exponential_stability_prelim}. Assume that the given system admits a continuously differentiable, positive definite function $\Psi\colon [0,\infty)\times D\to \mathbb{R}$, that satisfies the following:
\begin{align}\label{Eqn:polynomial_bound_Psi}
    k_1\|\xi\|^p\leqslant \Psi(t,\xi)\leqslant k_2\|\xi\|^p,
\end{align}
$\forall t\geqslant 0$ and $\forall \xi\in D$ where $k_1,k_2\in \mathbb{R}_+$ and $p$ is a positive integer. Further, if $\Psi(t,\xi)$ satisfies the
\begin{align}\label{Eqn:dot_Psi}
    \dot\Psi(t,\xi)\leqslant -\tilde k\Psi(t,\xi),\forall t\geqslant 0, \forall \xi\in D,
\end{align}
for some positive $\tilde k$, then the origin of system \eqref{Eqn:sys_exponential_stability_prelim} is exponentially stable. If \eqref{Eqn:polynomial_bound_Psi} and \eqref{Eqn:dot_Psi} hold for all $\xi\in \mathbb{R}^n$, then the origin is globally exponentially stable. 
\end{theorem}

\section{Unified barrier Lyapunov function}\label{Sec:Unified barrier Lyapunov function}
In this section, we propose a smooth \textit{logarithmic} asymmetric barrier Lyapunov function and construct class $\mathcal{K}_{\infty}$ functions that serve as upper bound and lower bound of the proposed function in its domain. 

Consider the open set $\mathcal{X}:=\{x\in \mathbb{R}: -a<x<b,~a,b\in \mathbb{R}_+ \}$ and define the function $V$ over the set $\mathcal{X}$, $V\colon \mathcal{X}\to \mathbb{R}_+\cup \{0\}$, which is expressed as
\begin{align}
V(x)\coloneqq \text{log}\!\left [\dfrac{a^a b^b}{(a+x)^a(b-x)^b}\right].\label{Eqn:proposed_ABLF}    
\end{align} 
It may be immediately seen that $V\rightarrow \infty$ as  $x\rightarrow -a$ or $x\rightarrow b$. Furthermore, it can also be concluded that for $a\ne b$, the function is asymmetric about $x=0$ since $V(-x)\ne V(x)$.  
\begin{remark}\label{Rem:irrational_powers}
    Since $a,b\in \mathbb{R}$, the proposed function $V(x)$ requires real-exponentiation. As $\mathbb{Q}$ (set of rational numbers) is dense in $\mathbb{R}$ \cite[Theorem 1.20]{rudin1976principles}, we can construct $a_q, b_q\in \mathbb{Q}$ arbitrarily close to $a$ and $b$ and satisfying $a_q<a$ and $b_q<b$. Thus, the function $V(x)$ is well defined even if $a$ and/or $b$ are irrational.
\end{remark}
\begin{lemma}\label{Lemma:pos_defi_C_inf_ABLF}
    The following hold for the function $V(x)$: 
    \begin{itemize}
        \item[a.] $V(x)$ is a class $\mathcal{C}^{\infty}$ function.
        \item[b.] $V(x)$ is strictly convex and positive definite over $\mathcal{X}$.
    \end{itemize}
\end{lemma}
\begin{proof}
\begin{itemize}
    \item[a.] The function $V(x)$, as given in \eqref{Eqn:proposed_ABLF}, can be rewritten as given below
    \small
    \begin{align}\label{Eqn:Va_plus_Vb}
    V(x)\!=V_a(x)\!+\!V_b(x)\!\coloneqq \!a\text{log}\!\left(\!\frac{a}{a\!+\!x}\! \right)\!+\!b\text{log}\!\left(\!\frac{b}{b\!-\!x} \!\right).
    \end{align}
    \normalsize
    The functions $V_a$  and $V_b$ are logarithms of rational polynomial functions and hence are $\mathcal{C}^{\infty}$ over $(-a,\infty)$ and $(-\infty,b)$ respectively. Therefore, $V(x)$ is $\mathcal{C}^{\infty}$ over $(-a,b)$ %as it is a sum of two $\mathcal{C}^{\infty}$ functions.
    \item[b.] The first and second derivative of $V(x)$ with respect to $x$ are:
    \begin{align}
        V^{\prime}(x)=\frac{\mathrm{d}V(x)}{\mathrm{d}x}&=\!\frac{(a\!+\!b)x}{(a\!+\!x)(b\!-\!x)},\label{Eqn:dV_dx}\\
        V^{\prime \prime}(x)=\frac{\mathrm{d}^2V(x)}{\mathrm{d}x^2}&=\!\frac{(a\!+\!b)(x^2 \!+\!ab)}{[(a\!+\!x)(b\!-\!x)]^2}.\label{Eqn:ddV_dxx}
    \end{align}
    Clearly $V^{\prime \prime}(x)>0,~\forall x\in \mathcal{X}$. Therefore, from the second-order convexity conditions \cite{boyd2009convex}, $V(x)$ is strictly convex. Furthermore, since $V(x)$ is defined over the convex set $\mathcal{X}$, which is open, $V(x)$ has at most one global minimum over $\mathcal{X}$ satisfying $V^{\prime}(x)=0$ \cite[Proposition 1.1.2]{bertsekas1999NonlinearProg}. Therefore, the function $V(x)$ has a global minimum at $x=0$. Further, we have $V(x)>V(0)~\forall x\in \mathcal{X}\backslash{\{0\}}$. Since $V(0)=0$, we have $V(x)>0~\forall x\in \mathcal{X}\backslash{\{0\}}$.% indicating that $V(x)$ is positive definite over $\mathcal{X}$.
    \vspace{-1.5mm}
\end{itemize}
\end{proof}

\begin{remark}\label{Rem:Choice_a_b}
Suppose the function $V(x)$, instead of being defined as in \eqref{Eqn:proposed_ABLF}, is alternatively defined as follows:
\begin{align}
    \widetilde V(x)\coloneqq \text{log}\!\left [\dfrac{a^{a_1} b^{b_1}}{(a+x)^{a_1}(b-x)^{b_1}}\right],\label{Eqn:alternate_proposed_ABLF}
\end{align}
where $a_1,b_1\in \mathbb{R}_+$ with $a_1\ne a$ and $b_1\ne b$. Taking the first derivative of $\widetilde V(x)$, as given in \eqref{Eqn:alternate_proposed_ABLF}, with respect to $x$ we have
\begin{align}
    \widetilde{V}^{\prime}(x)=\!\frac{(a_1\!+\!b_1)x+(b_1a-a_1b)}{(a\!+\!x)(b\!-\!x)}.\label{Eqn:alternate_dV_dx}
\end{align}
From \eqref{Eqn:alternate_dV_dx}, it is evident that the choice of $a_1$ and $b_1$ is crucial in determining the extremum of $V(x)$. It can be easily shown that $\widetilde V(x)$, as given in \eqref{Eqn:alternate_proposed_ABLF}, has an extremum at $x=0$ iff $a_1/b_1=a/b$. Furthermore, it can be easily verified that \textit{Lemma} \ref{Lemma:pos_defi_C_inf_ABLF} holds for arbitrary $a_1,b_1>0$ satisfying $a_1/b_1=a/b$.
% As a consequence of $V(x)$ being positive definite, the following relation holds:
% \begin{align}\label{Eqn:zeta_positive_def}
%     \zeta(x)\coloneqq \frac{a^a b^b}{(a+x)^a(b-x)^b}\geqslant 1,\forall~x\in \mathcal{X}.
% \end{align}
% Therefore, $\zeta(x)>0~\forall\ x\in\mathcal{X}$.

\end{remark} 
\begin{remark}\label{Rem:symm_constraints}
    Let $a=b$ in \eqref{Eqn:proposed_ABLF}, then we have
    \begin{align}\label{Eqn:a_equal_b_Sym_ABLF}
        V_{\text{sym}}(x)=\text{log}\left(\frac{a^a a^a}{(a^2-x^2)^a}  \right)=a ~\text{log}\left(\frac{a^2}{a^2-x^2}  \right).
    \end{align}
     The proposed function thus reduces to the symmetric barrier Lyapunov function in \cite{tee2009barrier}. Hence, the proposed function is a generalization of the existing logarithmic symmetric barrier Lyapunov function to the asymmetric case.
\end{remark}
\begin{remark}
The commonly used asymmetric barrier Lyapunov function, proposed in \cite{tee2009barrier}, is given by
\small
\begin{align}\label{Eqn:ABLF_KPT}
    V_{\text{exist}}(x)\!=\!\frac{1}{p}\left[ q(x)~\text{log}\frac{a^p}{a^p-x^p}\! +\!(1\!-\!q(x))~\text{log}\frac{b^p}{b^p-x^p}\right],
\end{align}
\normalsize
where $p$ is an even integer satisfying $p\geqslant 2$ and $q(x)$ is the switching function defined as
\begin{align}
    q(x)=\begin{dcases}
        1, &\text{if }x>0\\
        0, &\text{otherwise.}
    \end{dcases}\notag
\end{align}
Clearly, $V_{\text{exist}}(x)$ does not seamlessly reduce to the symmetric case from substituting $a=b$ unless $p=2$. Furthermore, the function, $V_{\text{exist}}(x)$, is $(p-1)$ times differentiable, which is due to the discontinuous function $q(x)$. Therefore, the higher the number of successive derivatives required in the controller design, the higher will be the required value of $p$. This is reflected in the control design in \cite{tee2009barrier}, where $p$ must satisfy $p\geqslant n$ with $n$ being the order of the system in \eqref{Eqn:sys_dynamics}.
\end{remark}
Following the results in \textit{Lemma} \ref{Lemma:pos_defi_C_inf_ABLF} we have the following conclusions on $V^{\prime \prime}(x)$.
\begin{corollary}\label{corollary:minima_of_ddot_V}
The second-order derivative of $V(x)$ with respect to $x$, $ V^{\prime \prime}(x)$, as given in \eqref{Eqn:ddV_dxx}, is strictly convex over $\mathcal{X}$ and admits a global minimum at $x=\sqrt[3]{b^2a}-\sqrt[3]{ba^2}\in \mathcal{X}$. Furthermore, the following statements hold for $V^{\prime \prime}(x)$: 
\begin{align}
    &\text{If}~~a<b~~\text{then}~~\left(\frac{b^2}{a}\right)V^{\prime \prime}(x)>1, \forall x\in \mathcal{X}.\label{Eqn:inequality_ddV_dxx_a<b}\\
    &\text{If}~~a>b~~\text{then}~~\left(\frac{a^2}{b}\right)V^{\prime \prime}(x)>1, \forall x\in \mathcal{X}.\label{Eqn:inequality_ddV_dxx_a>b}
\end{align}
\end{corollary}
\begin{proof}
Differentiating $V^{\prime \prime}(x)$ with respect to $x$, we have
\begin{align}
    V^{\prime \prime \prime}(x)&=(a+b)\frac{x^3+3ab+ab(a-b)}{(a+x)^3 (b-x)^3}\label{Eqn:third_derivative_V(x)},\\
    V^{(4)}(x)&=(a+b)\frac{(x^2+ab)^2+ab(2x+a-b)^2}{(a+x)^4 (b-x)^4}\label{Eqn:fourth_derivative_V(x)}.
\end{align}
From \eqref{Eqn:fourth_derivative_V(x)}, we see that $V^{(4)}(x)$ is positive definite and hence $V^{\prime \prime}(x)$ is strictly convex over $\mathcal{X}$. Therefore, there exists at most one global minimum of $V^{\prime \prime}(x)$ over $\mathcal{X}$ \cite[Proposition 1.1.2]{bertsekas1999NonlinearProg}. To find the minimum, we equate \eqref{Eqn:third_derivative_V(x)} to zero, from which we have
\begin{align}\label{Eqn:minima_for_ddV_dxx}
    x^3+3abx+ab(a-b)=0.
\end{align}
Hence, the minimum must be a root of the cubic polynomial \eqref{Eqn:minima_for_ddV_dxx}. We shall now apply Cardan's formula to find the roots. Comparing \eqref{Eqn:minima_for_ddV_dxx} with cubic polynomial \eqref{Eqn:depressed_cubic_polynomial}, we have $p=3ab$ and $q=ab(a-b)$ and $p,q\in \mathbb{R}$. Furthermore, we have $\Delta=4(3ab)^3+27a^2b^2(a-b)^2>0$ from which we conclude that \eqref{Eqn:minima_for_ddV_dxx} will admit one real root and two complex conjugate roots. Computing $A$ and $B$ from the equations in \eqref{Eqn:A_B_def_roots}, we have $A=ab^2$ and $B=-a^2b$. Hence, the real root for the polynomial \eqref{Eqn:minima_for_ddV_dxx} is given by $x=\sqrt[3]{ab^2}-\sqrt[3]{a^2b}$, which is the global minimum of $V^{\prime \prime}(x)$ over $\mathcal{X}$.

To prove the inequality in \eqref{Eqn:inequality_ddV_dxx_a<b}, it is enough to show that $(b^2/a)\underset{x\in \mathcal{X}}{\min} (V^{\prime \prime}(x))>1$. The minimum of $V^{\prime \prime}(x)$ occurs at $x=\sqrt[3]{ab^2}-\sqrt[3]{a^2b}$. Substituting for $x$ in \eqref{Eqn:ddV_dxx}, we have
\small
\begin{align}
    \underset{x\in \mathcal{X}}{\min} V^{\prime \prime}(x)\!=\!V^{\prime \prime}(\!\sqrt[3]{\!ab^2}\!-\!\sqrt[3]{\!a^2b})\!=\!\frac{(a\!+\!b)}{(b^{2/3}\!-\!(ab)^{1/3}\!+\!a^{2/3})^3}.
\end{align}
\normalsize
We normalize the above expression by considering $\tilde z\coloneqq b/a$. It is easy to see that $\tilde z>1$ as $b>a$. The normalized expression for $(b^2/a)\min(V^{\prime \prime}(x))$ is given by
\begin{align}\label{Eqn:normalized_min_ddV_dxx}
    (b^2/a)\min(V^{\prime \prime}(x))=\frac{2\tilde z^2(1+z)}{(\tilde z^{2/3}-\tilde z^{1/3}+1)^3}.
\end{align}
Suppose $(b^2/a)\min(V^{\prime \prime}(x))$ is not greater than 1, then from \eqref{Eqn:normalized_min_ddV_dxx} we have
\begin{align}
   \frac{\tilde z^2(1+\tilde z)}{(\tilde z^{2/3}-\tilde z^{1/3}+1)^3}\leqslant 1.
\end{align}
On performing further algebraic manipulations, we arrive at the following expression:
\begin{align}\label{Eqn:contradicting_inequality}
    \tilde z^{2/3}\Bigl((1+\tilde z)^{1/3}-1\Bigr)+(\tilde z^{1/3}-1)\leqslant0
\end{align}
This is a contradiction as $\Bigl((1+\tilde z)^{1/3}-1\Bigr)>0$ and $(\tilde z^{1/3}-1)>0$ since $\tilde z>1$. Therefore, we have $(b^2/a)\min(V^{\prime \prime}(x))>1$. Hence, the inequality in \eqref{Eqn:inequality_ddV_dxx_a<b} must hold. The proof for the case $a>b$ follows exactly as above with $\tilde z\coloneqq a/b$.
\end{proof}
As a consequence of the above result we have the following result for the case $a=b$.
\begin{corollary}\label{corollary:inequality_ddV_dxx_a_equal_b}
 When $a=b$, the following inequality holds:
\begin{align}\label{Eqn:inequality_ddV_dxx_a_equal_b}
    aV^{\prime \prime}(x)>1,\forall x\in \mathcal{X}
\end{align}
\end{corollary}
\begin{proof}
It is straightforward to see that \eqref{Eqn:contradicting_inequality} cannot hold even when $\tilde z=1$, i.e., when $a=b$. Hence \eqref{Eqn:inequality_ddV_dxx_a_equal_b} holds.
\end{proof}

We shall now state an important result pertaining to the proposed BLF, $V(x)$. This result will play a vital role in establishing the existence of functions that serve to bound the proposed function $V(x)$.
\begin{lemma}\label{lem:V(-x)_greater_than_V(x)}
    For the function $V(x)$ with $0<a<b$, as defined in \eqref{Eqn:proposed_ABLF}, the following holds:
    \begin{align}
        V(-x)>V(x)~\forall x\in(0,a)\subseteq \mathcal{X}.
    \end{align}
\end{lemma}
\begin{proof}
    For any arbitrary $x\in (0,a)$, the difference $\Delta_V=V(-x)-V(x)$ can be written as:
    \small
    \begin{align}
        \Delta_V&=\text{log}\!\left [\dfrac{a^a b^b}{(a-x)^a(b+x)^b}\right]-\text{log}\!\left [\dfrac{a^a b^b}{(a+x)^a(b-x)^b}\right]\notag\\
         &=a\! \left[\text{log}\!\left(1\!+\!\frac{x}{a} \right)\!-\!\text{log}\!\left(1\!-\!\frac{x}{a} \right)\!  \right]\!+\!b\! \left[\text{log}\!\left(1\!-\!\frac{x}{b} \right)\!-\!\text{log}\!\left(1\!+\!\frac{x}{b} \right)\!  \right]\notag
    \end{align}
    \normalsize
    Since $x<a<b$, we have $(x/a)<1$ and $(x/b)<1$. Therefore, we may expand the logarithmic terms using \eqref{Eqn:series_expand_log}. Expanding and simplifying further, we arrive at:
    \begin{align}\label{Eqn:Delta_V(-x)_greater_than_V(x)}
        \Delta_V=2\sum_{i=1}^{\infty}\frac{x^{2i+1}}{2i+1}\left(\frac{b^{2i}-a^{2i}}{(ab)^{2i}}\right).
    \end{align}
    The inifinite Since every term in the infinite series is greater than zero, from \eqref{Eqn:Delta_V(-x)_greater_than_V(x)}, we conclude that $\Delta_V>0~\forall x\in (0,a)$. Therefore, $V(-x)>V(x)$ over $(0,a)$.
\end{proof}
The above result indicates that $V(-x)$ bounds $V(x)$ from above in $(0,a)$. Furthermore, when $a>b>0$ the statement of \textit{Lemma} \ref{lem:V(-x)_greater_than_V(x)} changes to
\begin{align}
    V(x)>V(-x),\forall x\in (0,b).\notag
\end{align}
An example is illustrated in Fig.\ref{fig:V(-x)_greater_than_V(x)} with $a=3$ and $b=5$. 
\begin{figure}[h]
    \centering
    \includegraphics[width=0.85\linewidth]{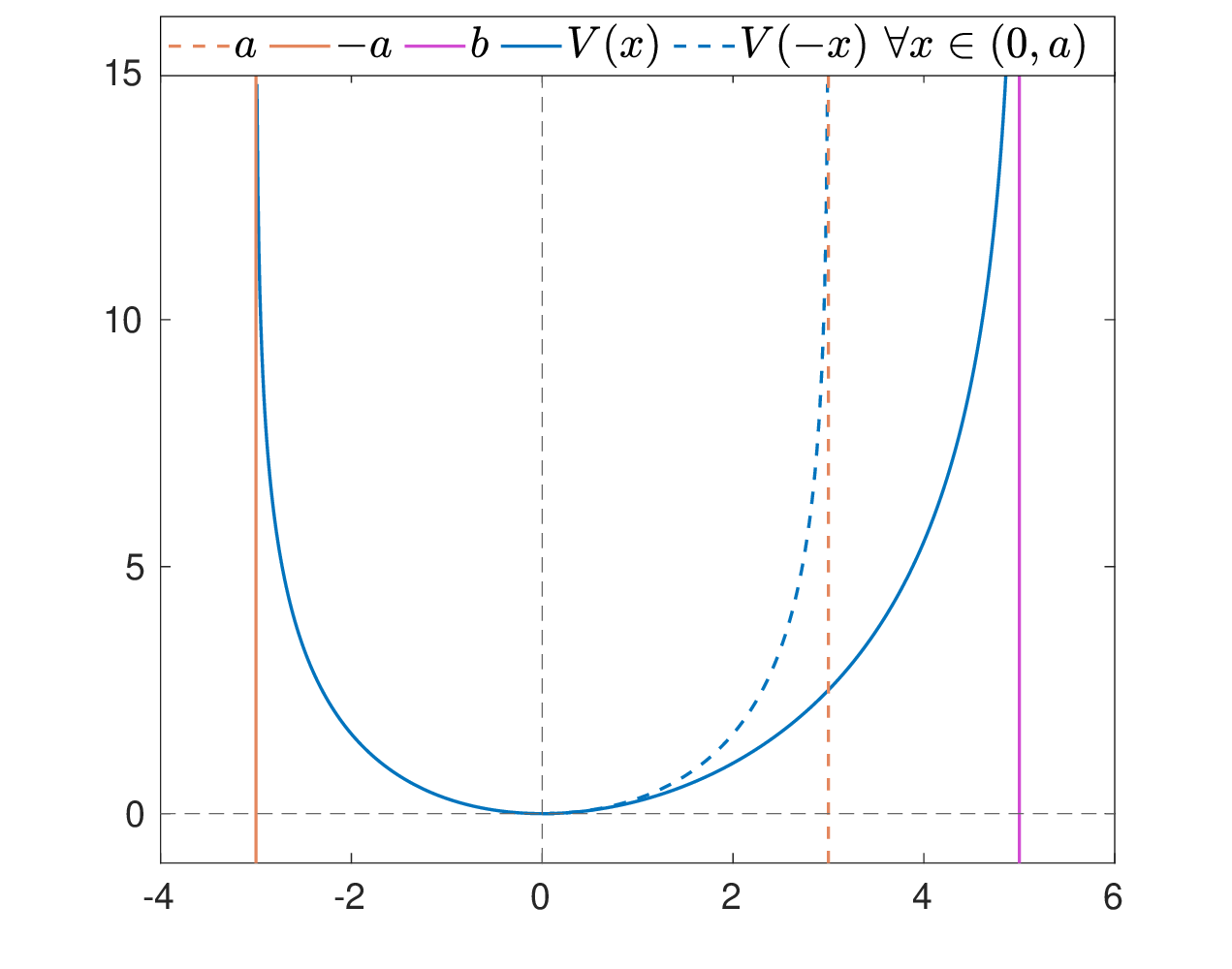}
    \caption{Illustration for \textit{Lemma} \ref{lem:V(-x)_greater_than_V(x)} with $a=3$ and $b=5$}
    \label{fig:V(-x)_greater_than_V(x)}
\end{figure}

Equipped with the preceding results, we shall now construct class $\mathcal{K}$ functions that bound the proposed BLF, $V(x)$ from above and below.
\subsection{Lower bounding class $\mathcal{K}$ function}
Without loss of generality, we show the existence of a lower bounding function for $V(x)$ when $a<b$. Similar conclusions will apply for the case of $a>b$.
\begin{lemma}\label{lemma:beta1_lower_bounding_func}
Consider the barrier Lyapunov function, $V(x)$, as defined in \eqref{Eqn:proposed_ABLF}, with $a<b$. Let $\beta_1(|x|)$ be a class $\mathcal{K}_{\infty}$ function defined as
\begin{align}\label{Eqn:beta_1_lower_bounding_fun}
    \beta_1(|x|)\coloneqq \Gamma_1 x^2, \text{ where }\Gamma_1\in \left (0,\frac{2}{a\!+\!b}\right).
\end{align}
Then $\beta_1(|x|)\leqslant V(x)~\forall x\in \mathcal{X}$.
\end{lemma}
\begin{proof}
From \textit{Lemma} \ref{lem:V(-x)_greater_than_V(x)}, $V(-x)> V(x)$ over $(0,a)$. Therefore, it is enough to show that $\beta_1(|x|)<V(x)$ over $(0,b)$ as this ensures $\beta_1(|x|)<V(x)$ over $\mathcal{X}\backslash\{ 0\}$.    

Define $\phi(x) \coloneqq V^{\prime}-\beta^{\prime}_1$, which can be written as
\begin{align}\label{Eqn:phi_beta1}
    \phi=\frac{(a+b)x}{(a+x)(b-x)}-2\Gamma_1 x, \forall x\in (0,b).
\end{align}
Note that $\phi(x)$ does not consider the argument $|x|$, this is because $x\in (0,b)$ and hence $x>0$. Furthermore, the denominator $(a+x)(b-x)$ has no roots over $(0,b)$ therefore, the function $\phi(x)$ is well defined. The polynomial $(a+x)(b-x)$ has a maximum at $x=(b-a)/2\in (0,b)$ with the maximum value given by $(b+a)^2/4$. Therefore, as $x>0$, we have
\begin{align}
    \phi(x)>\left(\frac{4}{(a+b)}-2\Gamma_1\right)x,~\forall x\in (0,b).
\end{align}
With $\Gamma_1<2/(a+b)$, we see that $\phi(x)>0~\forall x\in (0,b)$. We now have $V(x)=0,\beta_1(|x|)=0$ iff $x=0$ and $V^{\prime}(x)>\beta^{\prime}_1(|x|)~\forall x\in (0,b)$. With this, we may now write  
\begin{align}\label{Eqn:V_minus_beta1}
    V(x)-\beta_1(|x|)=\int^{x}_{0}\phi(\tau)d\tau,~\forall x\in [0,b).
\end{align}
As $\phi(x)>0$, we conclude $V(x)>\beta_1(|x|)~\forall x\in (0,b)$.
\end{proof}
\vspace{2mm}
\begin{remark}\label{Rem:beta_1_a_Greater_b}
    Note that, \textit{Lemma} \ref{lemma:beta1_lower_bounding_func} holds even when $a>b$. To see this, we may rewrite $V(x)-\beta_1(|x|)$ in \eqref{Eqn:V_minus_beta1} as
    \small
    \begin{align}
        V(x)\!-\!\beta_1(|x|)=\int^{0}_{x}\!(-\tau)\!\left[\! \frac{(a\!+\!b)}{(a\!+\!\tau)(b\!-\!\tau)}\!-\!2\Gamma_1\!\right]\!d \tau,\forall x\in\! (-a,0].
    \end{align}
    \normalsize
    With $\Gamma_1<2/(a+b)$, the integrand is greater than zero for all $x\in (-a,0)$ and therefore $V(x)>\beta_1(|x|)$.
\end{remark}
\begin{remark}\label{Rem:beta1_a_equal_b}
    Observe that \textit{Lemma} \ref{lemma:beta1_lower_bounding_func} holds even when $a=b$. Thus, when $\Gamma_1$ satisfies the condition $0<\Gamma_1<(1/a)$, we have $\beta_1(|x|)\leqslant V(x),~\forall x\in \mathcal{X}$. This essentially implies the existence of a lower bounding class $\mathcal{K}_{\infty}$ function for the standard logarithmic BLF in the case of symmetric constraints.  
\end{remark}
\subsection{Upper bounding class $\mathcal{K}$ function}
In the previous subsection, a lower bounding $\mathcal{K}$ function was shown to exist over $\mathcal{X}$. However, such global upper bounding class $\mathcal{K}$ function cannot be constructed as $V(x)\to \infty$ as $x\to -a$ or $x\to b$. In view of this, we construct class $\mathcal{K}$ functions that upper bound $V(x)$ over open subsets of $\mathcal{X}$.
\begin{lemma}\label{lemma:beta2_upper_bounding_func}
For the barrier Lyapunov function, $V(x)$, as defined in \eqref{Eqn:proposed_ABLF}, with $a<b$, given any $\tilde x\in\mathcal{X}$, there exists an open set ${\mathcal{S}}\coloneqq (-\bar a,\bar b)\subset \mathcal{X}$ with $\tilde x\in \mathcal{S}$, and a class $\mathcal{K}_{\infty}$ function $\beta_2(|x|)=\Gamma_2 x^2$ such that 
$V(x)\leqslant \beta_2(|x|),\forall x\in \mathcal{S}$.
\end{lemma}
\begin{proof}
As a consequence of \textit{Lemma} \ref{lem:V(-x)_greater_than_V(x)}, we can conclude that the construction for $\beta_2(|x|)$ must be guided by the lower of the two barriers $a$ and $b$. %Therefore, we first consider the case when $\tilde x\in(-a,0)$ since a symmetric upper bound function for $V(x)$ over $(-\bar a,0)$ is automatically an upper bound on it over $(0,\bar b)$.
Suppose $\tilde x\!=\!-k$ with $0\!<\!k\!<\!a$. Define the function, $\beta_2(|x|)$, $\beta_2\colon \mathbb{R}_+ \to \mathbb{R}_+$, as:
\small
\begin{align}\label{Eqn:beta_2_upper_bounding_fun}
    \beta_2(|x|)\coloneqq \Gamma_2x^2,~\Gamma_2\!=\!\frac{1}{k^2}\!\Bigg[ \text{log}\!\left (\!\dfrac{a^a b^b}{(a-k)^a(b+k)^b}\!\right)\!+\!\delta_{\beta_2}\!\Bigg],
\end{align}
\normalsize
where $\delta_{\beta_2}$ is an arbitrarily small positive constant chosen to ensure that $\tilde x=-k$ lies strictly within the interior of $\mathcal{S}$ as shown below. Consider the difference function $\Delta_{\beta_2}=\beta_2(|x|)-V(x)$ which can be rewritten as
\begin{align}\label{Eqn:Delta_beta_2}
   \Delta_{\beta_2}&=a \left[\text{log}\!\left(1+\frac{x}{a} \right)\!-\!\frac{x^2}{k^2}\text{log}\!\left(1-\frac{k}{a} \right)  \right]\notag\\
    &+b \left[\text{log}\!\left(1-\frac{x}{b} \right)\!-\!\frac{x^2}{k^2}\text{log}\!\left(1+\frac{k}{b} \right)  \right]\!+\!\delta_{\beta_2} \frac{x^2}{k^2}.
\end{align}
We have $k/a<1$ and considering $x\in (-a,a)$ we have $x/a<1$, $,x/b<1$. Thus expanding $\Delta_{\beta_2}$ using the logarithmic identity $\text{log}(1\!+\!\hat x)$ we have
% To see this, expand $\Delta_{\beta_2}$ for $x\in \!(-a,a)$ using the logarithmic identity $\text{log}(1\!+\!\hat x)$ as
\begin{align}\label{Eqn:Delta_beta_2_expanded}
\Delta_{\beta_2}=\sum_{j=1}^{\infty}&x^2\Bigg[\frac{k^{2j-1}+x^{2j-1}}{2j+1}\left(\frac{b^{2j}-a^{2j}}{a^{2j}~b^{2j}}\right)\notag\\
&+\frac{k^{2j}-x^{2j}}{2j}\left(\frac{b^{2j+1}+a^{2j+1}}{a^{2j+1} ~b^{2j+1}} \right)\Bigg]\!+\delta_{\beta_2} \frac{x^2}{k^2}.
\end{align}
Clearly, from \eqref{Eqn:Delta_beta_2_expanded} we have $\Delta_{\beta_2}>0$ when $x\in [-k,k]\backslash\{0\}$ (if $\delta_{\beta_2}=0$ we have $\Delta_{\beta_2}$=0 at $x=-k$). Further, the function $\Delta_{\beta_2}$ has zero crossings as $\beta_2(|x|)$ tends to infinity asymptotically and therefore will intersect the BLF, $V(x)$, at two points in $\mathcal{X}$ (see Fig \ref{fig:zero_crossing_for_beta_2} where $\Gamma_2$ is constructed for $x=-\tilde x_1,\tilde x_2$ and $x=-\tilde x_3$). Hence, there exist two points $x=-\bar a$ and $x=\bar b$, satisfying $k<\bar a<\bar b<b$, at which we have $\Delta_{\beta_2}=0$. It is easy to that $\bar a<\bar b$ holds as the inherent asymmetry of $V(x)$ leads $\beta_2(|x|)$ to intersect $V(x)$ at a point beyond $x=\bar a$.
%Notice that we claim $\bar a<\bar b$, this is due to the inherent asymmetry of $V(x)$ and the symmetry of $\beta_2(|x|)$, which leads $\beta_2(|x|)$ to intersect $V(x)$ at a point beyond $x=\bar a$. Therefore, we have $\Delta_{\beta_2}\geqslant 0,\forall x\in (-\bar a,\bar b)$, thus leading to $V(x)\leqslant \beta_2(\|x\|),\forall x\in (-\bar a,\bar b)$ and clearly $\tilde x=-k\in (-\bar a,\bar b)$.

Now, let $\tilde x\in (0,b)$ and suppose $\tilde x=k$. Let $-\hat{K}$ be the preimage of $V(k)$ in the open set $(-a,0)$, i.e., $\hat K=V^{-1}(V(k))$ over $(-a,0)$. Let $\Gamma_2$ now be given by
\begin{align}
\Gamma_2\!=\!\frac{1}{\hat K^2}\!\left[ \text{log}\left (\!\dfrac{a^a b^b}{(a-\hat K)^a(b+\hat K)^b}\!\right)\!+\!\delta_{\beta_2}\right].    
\end{align}
For the given $\Gamma_2$, from the previous arguments, we know that there exists $\mathcal{S}\!=\!(-\bar{a},\bar{b})$, with $\hat K\!<\!\bar{a}\!<\!\bar{b}<b$, over which $V(x)\!\leqslant\! \beta_2$. To argue that $\tilde x\!=\!k\in \mathcal{S}$, consider $x\!=\!\tilde \kappa$ to be the preimage of $V(-\bar a)$ over the open set $(0,b)$. The symmetry of $\beta_2$ prevents it from intersecting $V(x)$ at $x=\tilde \kappa$, which implies $\tilde \kappa<\bar b$ and thus $\tilde \kappa\in \mathcal{S}$. Furthermore, since $\hat{K}<\bar a$, we have $V(-\hat K)<V(-\bar a)$ as $V(x)$ is strictly convex and $V(x)\rightarrow \infty$ as $x\rightarrow -a$. Since $V(-\hat K)=V(k)$ and $V(-\bar a)=V(\tilde \kappa)$, we have $V(k)<V(\tilde \kappa)$ which further implies that $k<\tilde \kappa$ and therefore $k\in \mathcal{S}$.
\begin{figure}[htbp]
    \centering
    \includegraphics[scale=0.5]{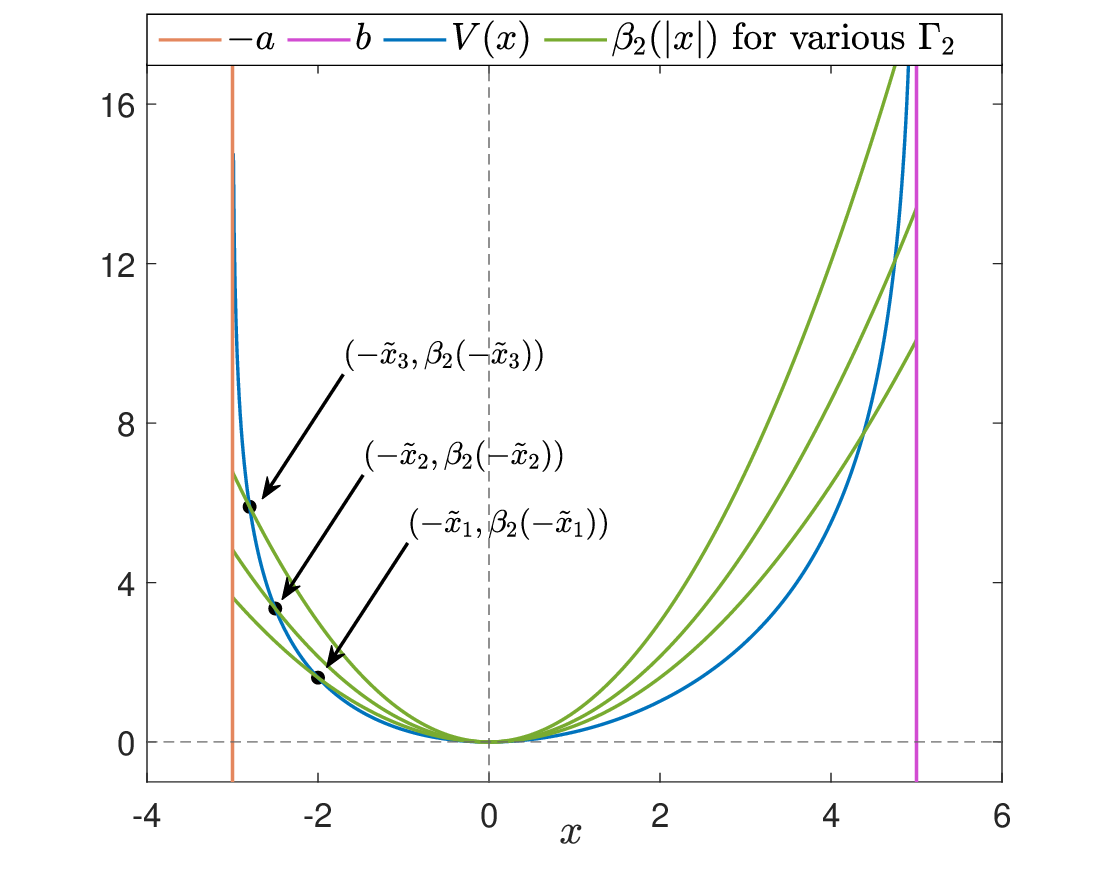}
    \caption{$\beta_2(|x|)$ with different $\Gamma_2$ constructed with $\tilde x_1\!<\!\tilde x_2\!<\!\tilde x_3$  and $\delta_{\beta_2}=0$}
    \label{fig:zero_crossing_for_beta_2}
\end{figure}
The construction of $\beta_2(|x|)$ when $a>b$ will follow along similarly by starting from the open set $(0,b)$.
\end{proof}
\begin{remark}\label{Rem:bounded_k2}
    Using the construction method above, we may construct $\beta_2(|x|)$ over $(-\bar a,\bar b)$ such that $\bar a=a-\epsilon_a$ and $\bar b=b-\epsilon_b$ with $\epsilon_a,\epsilon_b>0$ being arbitrarily small. Although $\Gamma_2$ grows large as $\epsilon_a,\epsilon_b\to 0$ it remains bounded as $x$ lies strictly within $(-a,b)$.
\end{remark}
%For the above method for the construction $\beta_2(|x|)$ we can find $\bar a$ and $\bar b$ that are arbitrarily close to $a$ and $b$, respectively, so that $\beta_2(|x|)$ bounds $V(x)$ above. Thus for any open set $\mathcal{S}\subset \mathcal{X}$ close to the set $\mathcal{X}$. 

%From \textit{Lemma} \ref{lemma:beta2_upper_bounding_func}, we see the construction of a class $\mathcal{K}_{\infty}$ upper bound of $V(x)$ for arbitrary open subsets of $\mathcal{X}$. We shall now look at a rational function that will bound the function $V(x)$ above $\forall x\in \mathcal{X}$.

\section{Control design and analysis} \label{Sec:Control_design_and_analysis}
In this section, using the proposed function, we shall first present the design steps, using the backstepping technique, to obtain the controller. Thereafter, we shall establish exponential stability for the closed-loop system.

\begin{remark}\label{Rem:Comment_on_backstepping}
    The design steps for the control follow along similarly as in \cite{tee2009barrier}. However, the designed stabilizing functions $\alpha_1$ and $\alpha_2$ differ from \cite{tee2009barrier}. Consequently, the steps involving $\alpha_1$ and $\alpha_2$ are described in detail and the rest of design is described in brief. 
\end{remark}
\textit{Step 1:} Define the error states as $z_1\coloneqq x_1\!-\!y_d$ and $z_2\coloneqq x_2\!-\!\alpha_{1}(z_1)$, where $\alpha_1$ is the stabilizing function we design in this step. From \textit{Assumption} \ref{Assumption:bounds_on_yd} and the output constraint $|y(t)|<k_y$, there exist constants $k_a>0$ and $k_b>0$ given by
\begin{align}
k_a=k_y-Y_{lb},~k_{b}=k_y-Y_{ub},\label{Eqn:ka_kb_def}    
\end{align}
and thus the error state $z_1$ satisfies
\begin{align}\label{Eqn:bounds_on_z1}
 -k_{a}<z_1<k_{b}.  
\end{align}
Therefore, $z_1$ belongs to the set $\mathcal{Z}_1\coloneqq \{ z_1\in \mathbb{R}\bigl \vert -k_a<z_1<k_b\}$. Consider the asymmetric barrier Lyapunov function, as given in \eqref{Eqn:proposed_ABLF}, defined over the set $\mathcal{Z}_1$:  
\begin{align}\label{Eqn:Design_V1_z1}
V_{1}(z_1)=\text{log}\!\left [\dfrac{k^{k_a}_a k^{k_b}_b}{(k_a+z_1)^{k_a}(k_b-z_1)^{k_{b}}}\right].
\end{align}
Differentiating $V_{1}$ with respect to time and substituting for $\dot z_1$ and $x_2$, we have
\begin{align}
\dot V_1\!=\!\dfrac{(k_a\!+\!k_b)z_1 \dot{z}_1 }{(k_a\!+\!z_1)(k_b\!-\!z_1)}\!=\!\dfrac{(k_a\!+\!k_b)z_1\bigl(f_1\!+\!g_1(z_2\!+\!\alpha_1)\!-\!\dot{y}_d\bigr) }{(k_a\!+\!z_1)(k_b\!-\!z_1)}.\notag
\end{align}
The stabilizing function, $\alpha_1$, is designed as:
\begin{align}\label{Eqn:alpha_1}
    \alpha_1\!=\!\frac{1}{g_1}\!\left(\!-f_1-\frac{\kappa_1 \tilde \kappa_1(k_a+k_b) z_1}{(k_a+z_1)(k_b-z_1)}+\dot{y}_d  \right),
\end{align}
where $\kappa_1\!>\!0$ is a constant and $\tilde \kappa_1$ is defined as: $\tilde \kappa_1\coloneqq (\max\{k_a,k_b\})^2/\min\{k_a,k_b \}$. Substituting \eqref{Eqn:alpha_1} in the expression for $\dot{V}_1$, we have the following relation:
\begin{align}
    \dot{V}_1=\frac{(k_a\!+\!k_b) g_1 z_1z_2}{(k_a\!+\!z_1)(k_b\!-\!z_1)}\!-\!\kappa_1\tilde \kappa_1\left[\!\frac{(k_a\!+\!k_b) z_1}{(k_a\!+\!z_1)(k_b\!-\!z_1)}\! \right]^2,
\end{align} 
%where $\tilde \beta(z_1)$ is as given in \eqref{Eqn:rational_upper_bound_a<b} or \eqref{Eqn:rational_upperBound_a>b} depending on the larger value among $k_a$ and $k_b$. 
where the cross term involving $z_1$ and $z_2$ will be eliminated in the next step.
\begin{remark}\label{Rem:difference_in_alpha1}
    When $k_a=k_b$, the stabilizing function $\alpha_1$ given by \eqref{Eqn:alpha_1} takes the following form:
    \begin{align}
        \alpha_1=\frac{1}{g_1}\!\left(\!-f_1-\frac{\kappa_1 2k^2_az_1}{k^2_a-z^2_1}+\dot{y}_d  \right),\label{Eqn:alpha1_for_equal_ka_kb}
    \end{align}
    which differs from  $\alpha_1(z_1)$ designed in \cite{tee2009barrier}. Consequently, the resulting control, $u(t)$, differs from the control derived for symmetric constraints in \cite{tee2009barrier}, despite the fact that $V(x)$, defined in \eqref{Eqn:Design_V1_z1}, reduces to the symmetric barrier Lyapunov function in \eqref{Eqn:a_equal_b_Sym_ABLF}.
\end{remark}
\textit{Step 2:}
Define $z_3\coloneqq x_3-\alpha_2$, where $\alpha_2$ is the stabilizing function to be designed in this step. State $x_2$ is unconstrained, a quadratic function is used for $z_2$. Let the Lyapunov candidate function be $V_2=V_1(z_1)+\frac{1}{2}z^2_2$.
% \begin{align}\label{Eqn:Design_V2}
%     V_2=V_1(z_1)+\frac{1}{2}z^2_2.
% \end{align}
The stabilizing function $\alpha_2$ is designed as
\begin{align}\label{Eqn:alpha_2}
        \alpha_2=\frac{1}{g_2}\left(-f_2-\frac{(k_a+k_b) g_1 z_1}{(k_a+z_1)(k_b-z_1)}-\kappa_2 z_2+ \dot{\alpha}_1 \right),
\end{align}
where $\kappa_2$ is a positive constant. Differentiating $V_2$ with respect to time and substituting for $\dot V_1$, $\dot x_2$, $x_3$ and $\alpha_2$, we have
% \begin{align}
% \dot{V}_2=\dot{V}_1+z_2\dot{z}_2=\dot{V}_1+z_2 \bigl(f_2+g_2(z_3+\alpha_2)-\dot{\alpha}_1\bigr).
% \end{align}
%  Substituting $\alpha_2$ in $\dot{V}_2$:
\begin{align}
    \dot{V}_2=-\kappa_1\tilde \kappa_1\left[\!\frac{(k_a\!+\!k_b) z_1}{(k_a\!+\!z_1)(k_b\!-\!z_1)}\! \right]^2 -\kappa_2z^2_2+g_2z_2z_3.
\end{align}
From Step 3 onwards, the design is carried out similarly by considering quadratic terms for each successive $z_i$. The general steps are as described below.

\textit{Step i $(3\leqslant i \leqslant n)$:} Define $z_{i+1}=x_{i+1}-\alpha_{i},\forall i\in \{3,4,\ldots,n \}$, where $\alpha_i$'s are the stabilizing functions designed at the $i^{\text{th}}$ step. Let the Lyapunov function candidate for \textit{step i} be ${V}_i=V_{i-1}+\frac{1}{2}z^2_{i}, ~~\forall i\in \{3,4,\ldots,n \}$. The corresponding stabilizing function, $\alpha_i$, is designed as 
\begin{equation}\label{Eqn:alpha_i}
    \begin{split}
    \alpha_{i}=\!\frac{1}{g_i}\!\bigl(\!-f_i\!-\!g_{i-1}z_{i-1}\!-\!\kappa_i z_i\!+\!\dot{\alpha}_{i-1} \bigr)~~\forall i\!\in\!\{3,\dots,n\}
    \end{split}
\end{equation}
where $\kappa_i$ is a positive constant. Note that the stabilizing function designed at \textit{Step $n$} is the control input $u(t)$. Therefore, the control $u(t)$ is given by
\begin{align}\label{Eqn:u(t)}
    u=\alpha_n=\frac{1}{g_n}\!\bigl(\!-f_n\!-\!g_{n-1}z_{n-1}\!-\!\kappa_n z_n\!+\!\dot{\alpha}_{n-1} \bigr).
\end{align}
The time derivative of $\dot \alpha_{i-1}$ is given as:
\begin{align}
    \dot{\alpha}_{i-1}=\sum^{i-1}_{j=1}\frac{\partial \alpha_{i-1}}{\partial x_j}\dot{x}_j+\sum^{i-1}_{j=0}\frac{\partial \alpha_{i-1}}{\partial y^{(j)}_d}y^{(j+1)}_d~~i\in \{2,3,\dots,n\},\notag
\end{align}
where $y^{(j)}_d$ is the $j^{\text{th}}$ order time derivative of $y_d$. 

Note: The stabilizing function $\alpha_1$ is a function of $z_1$ and $y_d$, similarly we have $\alpha_2=\alpha_2(x_1,x_2,z_1,z_2,y_d,\dot y_d)$. Along the same lines, the stabilizing function at \textit{Step} $i$ is a function of $x_i,z_i$ for $i\in \{1,2,\ldots,i\}$ and $y_d,\dot y_d,\ldots,y^{(i-1)}$. These explicit arguments have been omitted throughout this article for notational brevity.

This completes the design, and the closed-loop dynamics obtained after $n$ steps are given by:
\begin{equation}\label{Eqn:closed_loop_eqns}
    \begin{split}
        \dot{z}_1&=-\frac{\kappa_1 \tilde \kappa_1(k_a+k_b) z_1}{(k_a+z_1)(k_b-z_1)}+g_1z_2,\\
        \dot{z}_2&=-\dfrac{(k_a+k_b) g_1z_1}{(k_a+z_1)(k_b-z_1)}-\kappa_2z_2+g_2z_3,\\
        \dot{z}_i&=-g_{i-1}z_{i-1}-\kappa_i z_i+g_iz_{i+1},~~i=3,\dots,n-1\\
        \dot{z}_n&=-g_{n-1}z_{n-1}-\kappa_n z_n.
    \end{split}
\end{equation}
The time derivative of $V_{n}$ along the closed-loop system in \eqref{Eqn:closed_loop_eqns} can be written as
\begin{align}\label{Eqn:dot_Vn_1}
    \dot V_n=-\kappa_1\tilde \kappa_1\left[\!\frac{(k_a\!+\!k_b) z_1}{(k_a\!+\!z_1)(k_b\!-\!z_1)}\! \right]^2-\sum^n_{i=2}\kappa_iz^2_i.
\end{align}
\begin{remark}\label{Rem:stabilizing_func_C_infinity}
    Unlike in \cite{tee2009barrier}, where the stabilizing functions $\alpha_i$ are $\mathcal{C}^{n-i}$, the designed $\alpha_i$ in this article are $\mathcal{C}^{\infty}$.
\end{remark}
Note that the upper bounding and lower bounding quadratic class $\mathcal{K}$ functions already equips us with a tool to obtain exponential stability, provided a similar inequality can be shown for $\dot V_n(z)$. However, in that scenario the exponential decay rate will depend on $\Gamma_2$ which would result in slower decay rates for initial conditions very close to the boundary of $\mathcal{Z}_1$. Interestingly, it is easy to circumvent this issue and obtain an improved exponential stability result by showing that $V_{n}(z)$ satisfies the differential inequality in \eqref{Eqn:dot_Psi}. This is enabled by the following alternative upper bounding of $V_1(z_1)$.
%We shall now prove an important inequality that will enable us to establish the differential inequality in \eqref{Eqn:dot_Psi}.  %Clearly, the second result is a consequence of \textit{Lemma} \ref{lemma:beta1_lower_bounding_func} and \textit{Lemma} \ref{lemma:beta2_upper_bounding_func}.
\begin{lemma}\label{lemma:rational_func_upper_bound}
Consider the barrier Lyapunov function $V_1(z_1)$, as defined in \eqref{Eqn:Design_V1_z1}, and define the function  $\tilde \beta\colon \mathcal{Z}_1\longrightarrow \mathbb{R}$ as:
\begin{align}\label{Eqn:rational_upper_bound}
    \tilde \beta(z_1)\coloneqq\frac{(\max\{k_a,k_b \})^2}{\min\{k_a,k_b\}}\!\left[\!\frac{(k_a\!+\!k_b)z_1}{(k_a\!+\!z_1)(k_b\!-\!z_1)} \! \right]^2.
\end{align}
The following inequality holds: 
\begin{align}\label{Eqn:rational_func_inequality}
    V_1(z_1)\leqslant\tilde \beta(z_1),\forall z_1\in \mathcal{Z}_1.
    % &\text{If}~k_a>k_b~\text{then}~V_1(z_1)\leqslant \frac{k_a^2}{k_b}\left[\frac{(k_a+k_b)z_1}{(k_a+z_1)(k_b-z_1)}  \right]^2.
\end{align}
% then we have 
% \begin{align}
%     V(x)\leqslant \tilde \beta(x),~\forall x\in \mathcal{ X}.
% \end{align}
\end{lemma}
\begin{proof}
%Consider the inequality in \eqref{Eqn:rational_upper_bound_a<b}
We shall first consider the case for $k_a<k_b$. From \textit{Lemma} \ref{lem:V(-x)_greater_than_V(x)}, we have $V_1(z_1)>V_1(-z_1) \forall z_1\in (-k_a,0)$. Hence, it suffices to show that the inequality in \eqref{Eqn:rational_func_inequality} holds over $(-k_a,0)$. To show this, define $\tilde \phi(z_1) \coloneqq {\tilde \beta}^{\prime}(z_1)-V^\prime(z_1)$. Differentiating $\tilde \beta(z_1)$ and $V_1(z_1)$ with respect to $z_1$ and substituting in $\tilde \phi(z_1)$, we have
    \begin{align}
        \tilde{\phi}(z_1)=\frac{(k_a+k_b)z_1}{(k_a+z_1)(k_b-z_1)}\left[\frac{2k_b^2}{k_a} \frac{(k_a+k_b)(z^2_1+k_a k_b)}{(k_a+z_1)^2(k_b-z_1)^2}-1 \right].\notag
    \end{align}
    From \eqref{Eqn:ddV_dxx}, we can further reduce the equation for $\tilde \phi(z_1)$ as:
    \begin{align}\label{Eqn:phi_tilde_beta_tilde}
        \tilde{\phi}(z_1)=\frac{(k_a+k_b)z_1}{(k_a+z_1)(k_b-z_1)}\left[\frac{2k_b^2}{k_a}V^{\prime \prime}_1(z_1) -1 \right].
    \end{align}
    Similar to the arguments in the proof of \textit{Lemma} \ref{lemma:beta1_lower_bounding_func}, we may write $\tilde \beta (z_1)-V(z_1)$ as
    \begin{align}
        \tilde \beta (z_1)\!-\!V(z_1)&=\int_0^{z_1} \tilde{\phi}(\tau) d\tau,~\forall z_1\in\! (\!-k_a,0]\notag\\ 
        &{=\!\int_{z_1}^0 \!\frac{(k_a\!+\!k_b)(\!-\tau)}{(k_a\!+\!\tau)(k_b\!-\!\tau)}\!\left[\!\frac{2k_b^2}{k_a}V^{\prime\prime}_1(\tau)\!-\!1\! \right]\! d \tau.}\label{Eqn:tilde_beta_minus_V}
    \end{align}
    From \eqref{Eqn:inequality_ddV_dxx_a<b} in \textit{Corollary} \ref{corollary:minima_of_ddot_V}, we have $(k_b^2/k_a)V^{\prime \prime}_1(\tau)>1$ and therefore the integrand in \eqref{Eqn:tilde_beta_minus_V} is greater than zero for all $z_1\in (-k_a,0)$. Hence, we have $\tilde \beta(z_1)-V_1(z_1)>0\ \forall z_1\in (-k_a,0)$, from which we conclude $V_1(z_1)<\tilde \beta(z_1)$ over $\mathcal{Z}_1\backslash\{0 \}$.

    The proof for the case when $k_a>k_b$ follows along similar lines as above, using \eqref{Eqn:inequality_ddV_dxx_a>b} in \textit{Corollary} \ref{corollary:minima_of_ddot_V}.

    When $k_a=k_b$, the inequality in \eqref{Eqn:rational_func_inequality} becomes
\begin{align}\label{Eqn:rational_inequality_a_equal_b}
 k_a\text{log}\!\left[\!\frac{k_a^2}{k_a^2\!-\!z_1^2}\!\right]\!\leqslant k_a\!\left[\!\frac{2k_az_1}{k^2_a\!-\!z^2_1}\!\right]^2\!=\!\tilde \beta(z_1),\forall |z_1|<k_a,
\end{align}
which can be argued similarly through \textit{Corollary} \ref{corollary:inequality_ddV_dxx_a_equal_b}.
\end{proof}
\begin{remark}\label{Rem:Rational_upperBound_a_equal_b}
Recall the inequality from \cite{TeeRen_inequa_2010}
\begin{align}\label{Eqn:existing_inequality_tilde_beta}
    \text{log}\left[\frac{k_a^2}{k_a^2-z_1^2}\right]\leqslant \frac{z^2_1}{k^2_a-z^2_1},\forall |z_1|<k_a
\end{align} 
A brief analysis reveals that the following holds:
\begin{align}\label{Eqn:overall_compar_inequality}
    \text{log}\!\left[\!\frac{k_a^2}{k_a^2\!-\!z_1^2}\!\right]\!\leqslant\! \frac{z^2_1}{k^2_a\!-\!z^2_1}\!\leqslant\! \left[\!\frac{2k_az_1}{k^2_a\!-\!z^2_1}\!\right]^2,\forall |z_1|<k_a.
\end{align}
From \eqref{Eqn:overall_compar_inequality}, we see that the proposed inequality \eqref{Eqn:rational_func_inequality} extends the inequality \eqref{Eqn:existing_inequality_tilde_beta} from symmetric to asymmetric bounds.
\end{remark}
We shall now establish that $V_n(z)$ can be polynomially bound as in \eqref{Eqn:polynomial_bound_Psi}. 
\begin{corollary}\label{corollary:bounds_Vn}
For arbitrary $z_1\in \mathcal{Z}_1$, there exists $\mathcal{S}\subset \mathcal{Z}_1$, and positive constants $k_1,k_2\in\mathbb{R}$ such that $z_1\in\mathcal{S}$ and
\begin{align}\label{Eqn:bounds_Vn}
    k_1\|z\|^2\leqslant V_n(z)\leqslant k_2\|z\|^2~~~\forall z\in \mathcal{S}\times \mathbb{R}^{n-1}
\end{align}
\end{corollary}
\begin{proof}
    From \textit{Lemma} \ref{lemma:beta1_lower_bounding_func}, given $\Gamma_1\in (0,2/(k_a+k_b))$ we have $\Gamma_1 z^2_1\leqslant V_1(z_1)$ over $\mathcal{Z}_1$. Therefore we have
    \small
    \begin{align}\label{Eqn:lower_bound_Vn}
        V_n(z)=V_1(z_1)\!+\!\frac{1}{2}\sum_{i=2}^n z^2_i\geqslant\!\left(\!\Gamma_1 z^2_1\!+\!\frac{1}{2}\sum_{i=2}^n z^2_i\right)\!\geqslant k_1\|z\|^2,
    \end{align}
    \normalsize
    with $k_1=\min\{\Gamma_1,(1/2)\}$ for all $z\in \mathcal{Z}_1\times \mathbb{R}^{n-1}$. Further, from \textit{Lemma} \ref{lemma:beta2_upper_bounding_func} for arbitrary $z_1\in \mathcal{Z}_1$, there exists $\mathcal{S}$ and $\beta_2(|z_1|)=\Gamma_2 z^2_1$ such $\beta_2(|z_1|)\leqslant V_1(z_1)$ over $z\in \mathcal{S}$ and $z_1\in\mathcal{S}$. Therefore, we may now write
    \small
    \begin{align}\label{Eqn:upper_bound_Vn}
        V_n(z)=V_1(z_1)\!+\!\frac{1}{2}\sum_{i=2}^n z^2_i\leqslant  \!\left(\!\Gamma_2 z^2_1\!+\!\frac{1}{2}\sum_{i=2}^n z^2_i\right)\!\leqslant k_2\|z\|^2,
    \end{align}
    \normalsize
    for all $z\in \mathcal{S}\times \mathbb{R}^{n-1}$, where $k_2=\max\{\Gamma_2,(1/2) \}$. Hence from \eqref{Eqn:lower_bound_Vn} and \eqref{Eqn:upper_bound_Vn} we have \eqref{Eqn:bounds_Vn}.% Clearly we have $\mathcal{S}\times \mathbb{R}^{n-1}\subset \mathcal{Z}_1\times \mathbb{R}^{n-1}$, 
\end{proof}
Equipped with the two results above, we are now ready to carry on stability analysis for the designed closed-loop system.% in \eqref{Eqn:closed_loop_eqns}
\begin{theorem}\label{theorem:local_exponential_stability}
Consider the closed-loop system, as given in \eqref{Eqn:closed_loop_eqns}, designed under \textit{Assumption} \ref{Assumption:bounds_on_yd} and \textit{Assumption} \ref{Assumption:bounds_on_gi}. Starting from an initial condition $z_0 \coloneqq [z_{1_0},z_{({2:n})_0}]^T\in\Omega_{0}\coloneqq \mathcal{Z}_1\times \mathbb{R}^{n-1}$, the closed-loop system \eqref{Eqn:closed_loop_eqns} converges to the origin exponentially fast while $z_1(t)$ satisfies $z_1(t)\in\mathcal{Z}_1,\forall t\geqslant 0$.
\end{theorem}
\begin{proof}
Consider the Lyapunov function candidate $V_n(z)$:
\begin{align}
    V_n(z)=V_1(z_1)+\sum_{i=2}^{n}\frac{1}{2}z^2_i.\notag
\end{align}
The function $V_n(z)$ is positive definite and continuously differentiable over $\Omega_0$. From \textit{Corollary} \ref{corollary:bounds_Vn}, there exists domain $D=\mathcal{S}\times \mathbb{R}^{n-1}\subset \Omega_0$ containing the origin $z=0$ such that
\begin{align}\label{Eqn:bounds_Vn_2}
    k_1\|z\|^2\leqslant V_n(z)\leqslant k_2\|z\|^2~~~\forall z\in D.
\end{align}
Furthermore, from the definition of $\tilde \beta(z_1)$ in \eqref{Eqn:rational_upper_bound}, the time derivative of $V_n(z)$, as given in \eqref{Eqn:dot_Vn_1}, along the closed-loop trajectories of \eqref{Eqn:closed_loop_eqns} may be further written as
\begin{align}
    \dot V_n(z)=-\kappa_1\tilde\beta(z_1)-\sum_{i=2}^n\kappa_i z^2_i.
\end{align}
From \eqref{Eqn:rational_upper_bound}, it is clear that $\tilde \beta(z_1)>0\ \forall z_1\in \mathcal{Z}_1\backslash \{ 0\}$. Hence, we have $\dot V_n(z)<0$ in $\Omega_0\backslash\{\mathbf{0}\}$. Thus, from \textit{Lemma} \ref{lemma:lemma_1_KPT}, we conclude that $z_1(t)$ satisfies $z_1\in \mathcal{Z}_1,\forall t\geqslant 0$. Now, from the inequality \eqref{Eqn:rational_func_inequality} of \textit{Lemma} \ref{lemma:rational_func_upper_bound} we may rewrite $\dot V_n(z)$ as
\begin{align}\label{Eqn:dot_Vn_with_tilde_beta}
    &\dot V_n(z)<-\kappa_1V_1(z_1)-\sum_{i=2}^n\kappa_i z^2_i\leqslant -\tilde k V_n(z),
\end{align}
$\forall z\in D$ and $ \forall t\geqslant 0$ with $\tilde k=\min\{\kappa_1, 2\kappa_2,\ldots,2\kappa_n \}$ and $\kappa_i$ as described in the design steps. Thus, from \textit{Theorem} \ref{theorem:exponetial_stability}, the origin, $z=0$, of the closed-loop system \eqref{Eqn:closed_loop_eqns} is exponentially stable.
From \eqref{Eqn:dot_Vn_with_tilde_beta}, we now have
\begin{align}\label{Eqn:Vn(z)_less_than_Vn(z0)}
    V_n(z)\leqslant V_n(z_0)e^{-\tilde k t},\forall z_0\in D.
\end{align}
From \eqref{Eqn:bounds_Vn_2}, we have $k_1\|z\|^2\leqslant V_n(z)\leqslant V_n(z_0)e^{-\tilde k t}$ from which we may write 
\begin{align}
    \|z\|\leqslant \left(\!\frac{V_n(z_0)}{k_1} \!\right)^{\!1/2}e^{-\tilde k t/2}.\notag
\end{align}
Furthermore, from \eqref{Eqn:bounds_Vn_2}, we may write $V_n(z_0)<k_2\|z_0\|^2$ for any arbitrary $z_{1_0}\in \mathcal{S}$. Hence, we have %the following inequality:
\begin{align}\label{Eqn:bound_on_norm(z)}
    \|z\|\leqslant \left(\!\frac{k_2}{k_1} \!\right)^{\!1/2}\|z_0\|e^{-\tilde k t/2}.
\end{align}
From \eqref{Eqn:Vn(z)_less_than_Vn(z0)} and \eqref{Eqn:bounds_Vn_2}, we have the following inequalities
\begin{align}
    &\Gamma_1 z^2_1<V_1(z_1)\leqslant V_n(z_0)e^{-\tilde k t}\leqslant k_2\|z_0\|^2 e^{-\tilde k t},\notag\\
    &\sum_{i=2}^n\frac{1}{2}z_i^2\leqslant V_n(z_0)e^{-\tilde k t}\leqslant k_2\|z_0\|^2 e^{-\tilde k t},
\end{align}
from which, we have the following bounds on $z_1$ and $z_{2:n}$:
\begin{align}\label{Eqn:expo_bound_on_z1(t)}
    |z_1|\leqslant \sqrt{\!\frac{k_2}{\Gamma_1} }~\|z_0\|e^{-\tilde k t/2},\|z_{2:n}\|\leqslant \sqrt{2k_2}~\|z_0\|e^{-\tilde kt/2}.
\end{align}
% Similarly, from \eqref{Eqn:Vn(z)_less_than_Vn(z0)} and \textit{Corollary} \ref{corollary:bounds_Vn} we may also write
% \begin{align}
%     \sum_{i=2}^n\frac{1}{2}z_i^2\leqslant V_n(z_0)e^{-\tilde k t}\leqslant k_2\|z_0\|^2 e^{-\tilde k t}, 
% \end{align}
% which further implies
% \begin{align}
    
% \end{align}
\end{proof}
\begin{remark}\label{Rem:all_about_symmetric_case}
Notice that \textit{Corollary} \ref{corollary:bounds_Vn} holds for the symmetric case $k_a=k_b$, as we have established the construction of upper and lower bounding class $\mathcal{K}$ functions for $V_1(z_1)$ with symmetric constraints in \textit{Section} \ref{Sec:Unified barrier Lyapunov function}. Furthermore, the inequality in \eqref{Eqn:rational_func_inequality} also holds, as shown in \textit{Remark} \ref{Rem:Rational_upperBound_a_equal_b}. Thus, the proposed design guarantees exponential stability even in the case of symmetric constraints.
%along with constructions in \textit{Section} \ref{Sec:Unified barrier Lyapunov function} and \textit{Lemma} \ref{lemma:rational_func_upper_bound} generalizes the exponential stability guarantees for the case of symmetric constraints.
\end{remark}
\begin{theorem}\label{theorem:bounds_zi_and_y(t)}
    For the given closed-loop system in \eqref{Eqn:closed_loop_eqns} the following statements hold:
    \begin{itemize}
        \item[(i)] Let $V_{n_0}$ be Lyapunov function $V_n(z)$ evaluated at $z_0$. Define positive constants $K_a$ and $K_b$ as
        \begin{align}
            &K_a\coloneqq k_a(1\!-\!\gamma^{1/k_a}_a), \gamma_a=\left[\!\frac{k_a}{k_a\!+\!k_b}\!\right]^{\!k_b}\!e^{-V_{n_0}}\\
            &K_b\coloneqq k_b(1-\gamma^{1/k_b}_b), \gamma_b=\left[\!\frac{k_b}{k_a\!+\!k_b}\!\right]^{\!k_a}\!e^{-V_{n_0}}.
        \end{align}
        The error state $z_1(t)$ is bounded by
        \begin{align}
            &-K_a\!\leqslant \!z_1(t)\!\leqslant\!K_b,\ \forall t\geqslant0.
        \end{align}
        \item[(ii)] All the stabilizing functions $\alpha_1,\alpha_2,\ldots,\alpha_n$ are bounded.
        \item[(iii)] The output constraint is not violated, i.e., $|y(t)|<k_y,\forall ~t\geqslant 0$.
    \end{itemize}
\end{theorem}
\begin{proof}
    \begin{itemize}
        \item[(i)] Without loss of generality, considering the initial time to be $t=0$, we may write \eqref{Eqn:Vn(z)_less_than_Vn(z0)} as
        \begin{align}
         V_1(z_1(t))\leqslant V_{n_0}e^{-\tilde k t}\leqslant V_{n_0},   
        \end{align}
        since $\tilde k>0$. Thus, we have 
        \begin{align}\label{Eqn:overal_inequality_for_Ka_Kb}
            \left(\!\frac{k_a}{k_a\!+\!z_1}\!\right)^{\!k_a}\left(\!\frac{k_b}{k_b\!-\!z_1}\!\right)^{\!k_b}\leqslant e^{V_{n_0}}.
        \end{align}
        Furthermore, the function $k_b/(k_b-z_1)$ is monotonically increasing in $(-\infty,0)$ and hence when $z_1\in (-k_a,0)$, we have $\bigl(k_b/(k_b+k_a)\bigr)<\bigl(k_b/(k_b-z_1)\bigr)<1$. We may thus use \eqref{Eqn:overal_inequality_for_Ka_Kb} to conclude the following: 
        \begin{align}
            \left(\!\frac{k_a}{k_a\!+\!z_1}\!\right)^{\!k_a}\!\left(\!\frac{k_b}{k_b\!+\!k_a}\!\right)^{\!k_b}\!\leqslant e^{V_{n_0}},\forall z_1\in (-k_a,0].\label{Eqn:new_bound_z1_ka}
        \end{align}
        Similarly, the function $k_a/(k_a+z_1)$ is monotonically decreasing in $(0,\infty)$ and hence when $z_1\in (0,k_b)$, we have $1>\bigl(k_a/(k_a+z_1)\bigr)>\bigl(k_a/(k_a+k_b)\bigr)$. So, we have
        \begin{align}
            \left(\!\frac{k_a}{k_a\!+\!k_b}\!\right)^{\!k_a}\!\left(\!\frac{k_b}{k_b\!-\!z_1}\!\right)^{\!k_b}\!\leqslant\! e^{V_{n_0}},\forall z_1\in [0,k_b).\label{Eqn:new_bound_z1_kb}
        \end{align}
        Thus, from \eqref{Eqn:new_bound_z1_ka} and \eqref{Eqn:new_bound_z1_kb}, we have $-K_a\!\leqslant\! z_1(t)\!\leqslant\! K_b,\forall t>0$.% Further, the bound on remaining error states $\|z_{2:n}\|\!\leqslant\!\sqrt{2V_{n0}}~e^{-\tilde k t/2}$ is straightforward to see since we have $\sum_{i=2}^n(z^2_i(t)/2)\leqslant {V_{n0}e^{-\tilde k t}}$ from \eqref{Eqn:Vn(z)_less_than_Vn(z0)}.
        \item[(ii)] In (i) we showed that $z_1,z_2,\ldots,z_n$ are bounded. Further, from \textit{Assumption} \ref{Assumption:bounds_on_yd}, the desired trajectory $y_d(t)$ and its successive derivatives are bounded. The boundedness of $z_1$ and $y_d$ ensures the boundedness of the state $x_1$. Since $\dot y_d$ is bounded and $f_1$ is a function of $x_1$, we may conclude that $\alpha_1$, defined in \eqref{Eqn:alpha_1}, is also bounded. The boundedness of $\alpha_1$ leads to the boundedness of $x_2$ as $x_2=z_2+\alpha_1$. This further implies boundedness of $\alpha_2$ (defined in \eqref{Eqn:alpha_2}), since it is a continuous function of $x_i(t),z_i(t),~i\in \{1,2\}$ and $\dot{y}_d(t),~\ddot{y}_d(t)$ which are all bounded. Since $x_3=z_3+\alpha_2$, we may now conclude that $x_3$ is bounded. Proceeding similarly, we can show that all the stabilizing functions $\alpha_3,\ldots,\alpha_n$ are bounded. Since $\alpha_n=u$, we conclude that the control input is also bounded.
        \item[(iii)] From the definition of the error state, $z_1(t)=x_1(t)-y_d(t)$, and the output equation, $y(t)=x_1(t)$, we have $y(t)=z_1(t)+y_d(t)$. From (i) we have $-K_a\leqslant z_1\leqslant K_b$, and from \textit{Assumption} \ref{Assumption:bounds_on_yd} we have $-Y_{lb}<y_d(t)<Y_{ub}$. Thus, from the relations discussed, we have:
        \small
        \begin{align}
            -(k_a\!+\!Y_{lb})\!<\!-(K_a\!+\!Y_{lb})\!<y(t)\!<\!K_b\!+\!Y_{ub}\!<\!k_b\!+\!Y_{ub}.\notag
        \end{align}
        \normalsize
        From \eqref{Eqn:ka_kb_def}, we may write $k_a+Y_{lb}=k_b+Y_{ub}=k_y$. Therefore, we  conclude that $|y(t)|<k_y,\forall ~t\geqslant 0$.
    \end{itemize}
\end{proof}
\begin{remark}\label{Rem:tuning_of_kappa1_alpha1}
    For highly asymmetric output constraints where $k_a\ll k_b$ or $k_a\gg k_b$, the control $u(t)$ can become very large due to $\tilde \kappa_1$ introduced in $\alpha_1$ (refer \eqref{Eqn:alpha_1}). This can be overcome by choosing a suitable control gain $\kappa_1$ or by scaling $\tilde \kappa_1$ by a small gain $\epsilon(\tilde \kappa_1)$. Such a scaling reduces the rate of convergence $\tilde k$ but does not alter the preceding analysis.
\end{remark}
\section{Numerical Examples}
We present two numerical examples. First, we shall consider the second-order system in \cite{tee2009barrier} for both the symmetric case: $k_a=k_b$, and the asymmetric case: $k_a\ne k_b$ with $k_a>k_b>0$. We then consider a third-order system and demonstrate scalability in terms of the order of the system. In the second example, we consider the case when $k_a<k_b$.

Consider the second-order system given by
\begin{align}
\begin{split}
    &\dot x_1=0.1 x^2_1+x_2\\
    &\dot x_2=0.1x_1 x_2-0.2x_1+(1+x^2_1)u,~~y=x_1.
\end{split}
\end{align}
The aim is to track the desired trajectory given by $y_d=0.2+0.3\sin{t}$, subject to the output constraint $|x_1|<k_y=0.56$. The desired output trajectory is bounded by: $-0.1<y_d(t)<0.5$, and from \eqref{Eqn:ka_kb_def}, we obtain $k_a=0.46$ and $k_b=0.06$.

Starting from $x_1(0)=-0.14$ and $x_2(0)=1.5$, with control gains $\kappa_1=\kappa_2=1$, the output trajectory for the proposed control $u(t)$ based on the unified ABLF and the ABLF in \cite{tee2009barrier} is shown in Fig.~\ref{fig:x1_ABLF_comparision}. 
\begin{figure*}[h!]
    \centering
    \begin{subfigure}{0.45\textwidth}
        \centering 
        \includegraphics[scale=0.45]{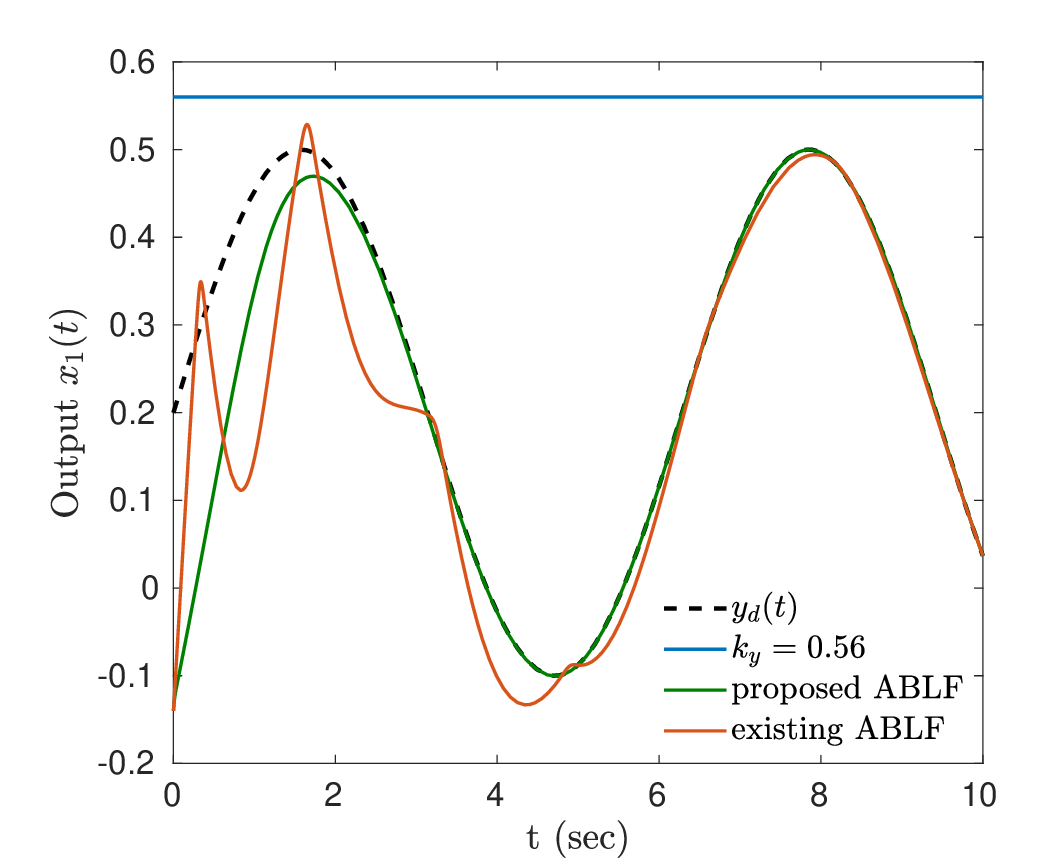}
        \caption{Output tracking for controllers based on the ABLF in \cite{tee2009barrier} and the proposed ABLF}
        \label{fig:x1_ABLF_comparision}
    \end{subfigure}
    \begin{subfigure}{0.45\textwidth}
        \centering
        \includegraphics[scale=0.48]{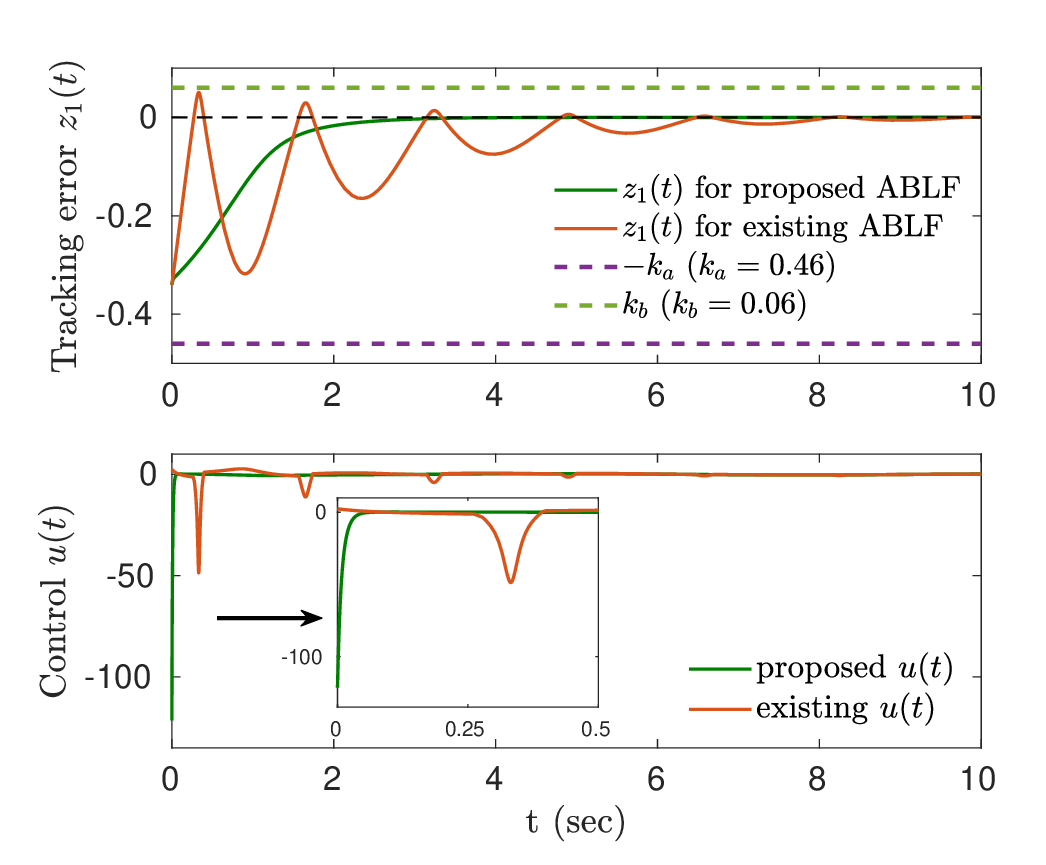}
        \caption{Tracking error $z_1(t)$ and control $u(t)$}
        \label{fig:z1_and_u(x10_m_point_14)}
    \end{subfigure}
    \caption{Illustration of performance for existing control and the proposed control}
    \label{fig:comparision_case}
\end{figure*}
We see that the state $x_1(t)$ converges to the desired trajectory while ensuring that the constraint $|x_1|<0.56$ is never violated. Fig.\ \ref{fig:z1_and_u(x10_m_point_14)} shows the corresponding tracking error $z_1(t)$ and control $u(t)$. Observe that the proposed $u(t)$ has a larger peak compared to the existing $u(t)$ (about 2.5 times larger as observed in the numerical simulation) correspondingly we see that $z_1(t)$ converges faster (about 3 times faster). Numerical simulations for initial conditions $x_1(0)\in \{-0.2,-0.08,0.12,0.24\}$ (correspondingly we have $z_{1_0}\in \{-0.4,-0.28,-0.08,0.04 \}$) and control gains $\kappa_1\!=\!\kappa_2\!=\!2$ are shown in Fig.\ \ref{fig:z1_asymmetric}. From Figs.\ \ref{fig:z1_asymmetric_proposed} and \ref{fig:z1_asymmetric_KPT}, we see that starting from any arbitrary $x_1(0)$ such that $z_{1_0}\in (-0.46,0.06)$, the performance of the proposed controller fares better than the existing control scheme proposed in \cite{tee2009barrier}. 
\begin{figure*}[h]
\centering
    \begin{subfigure}{0.49\textwidth}
        \centering
        \includegraphics[scale=0.43]{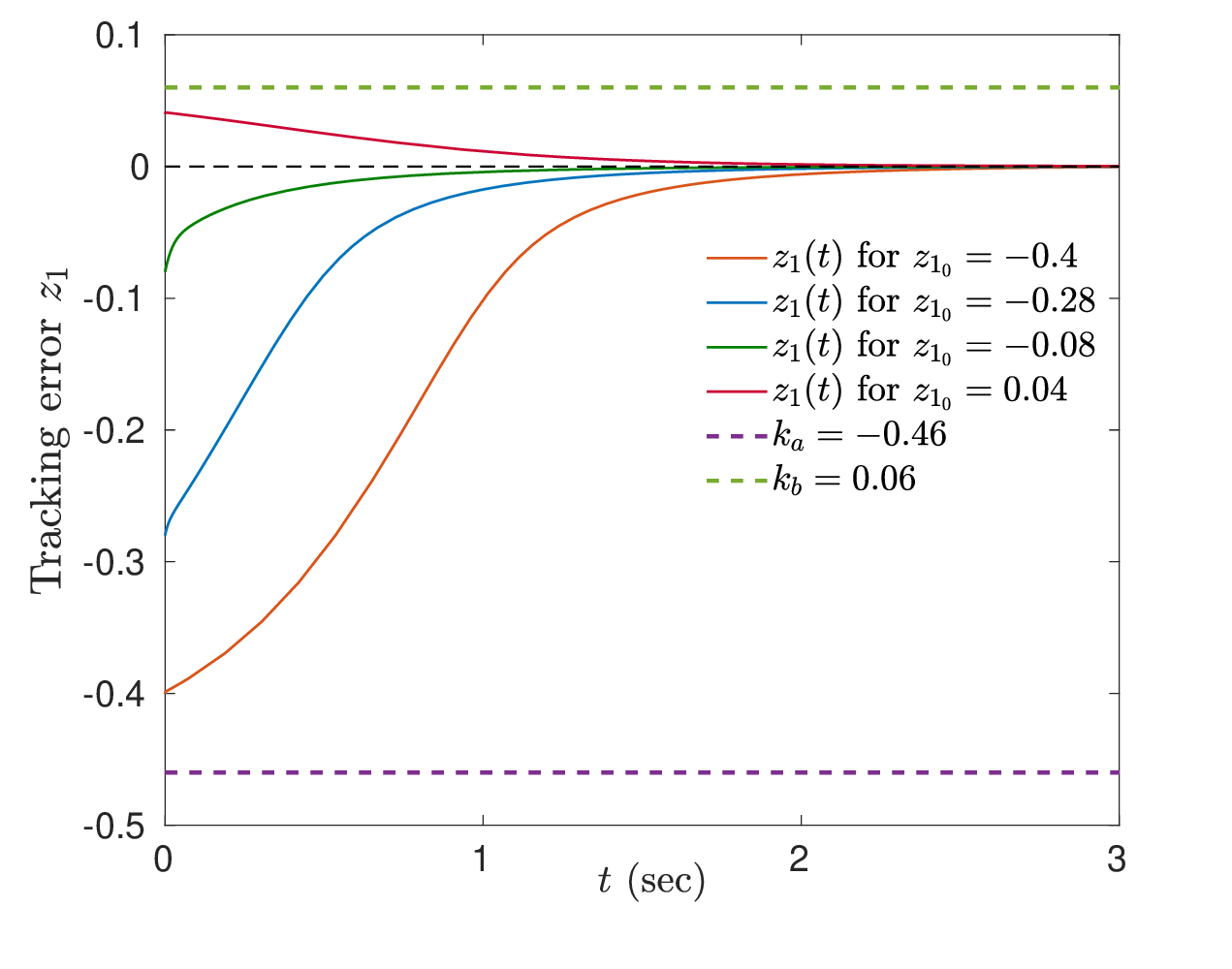}
        \caption{Tracking error $z_1(t)$ based on proposed ABLF}
        \label{fig:z1_asymmetric_proposed}
    \end{subfigure}
    \begin{subfigure}{0.44\textwidth}
        \centering
        \includegraphics[scale=0.43]{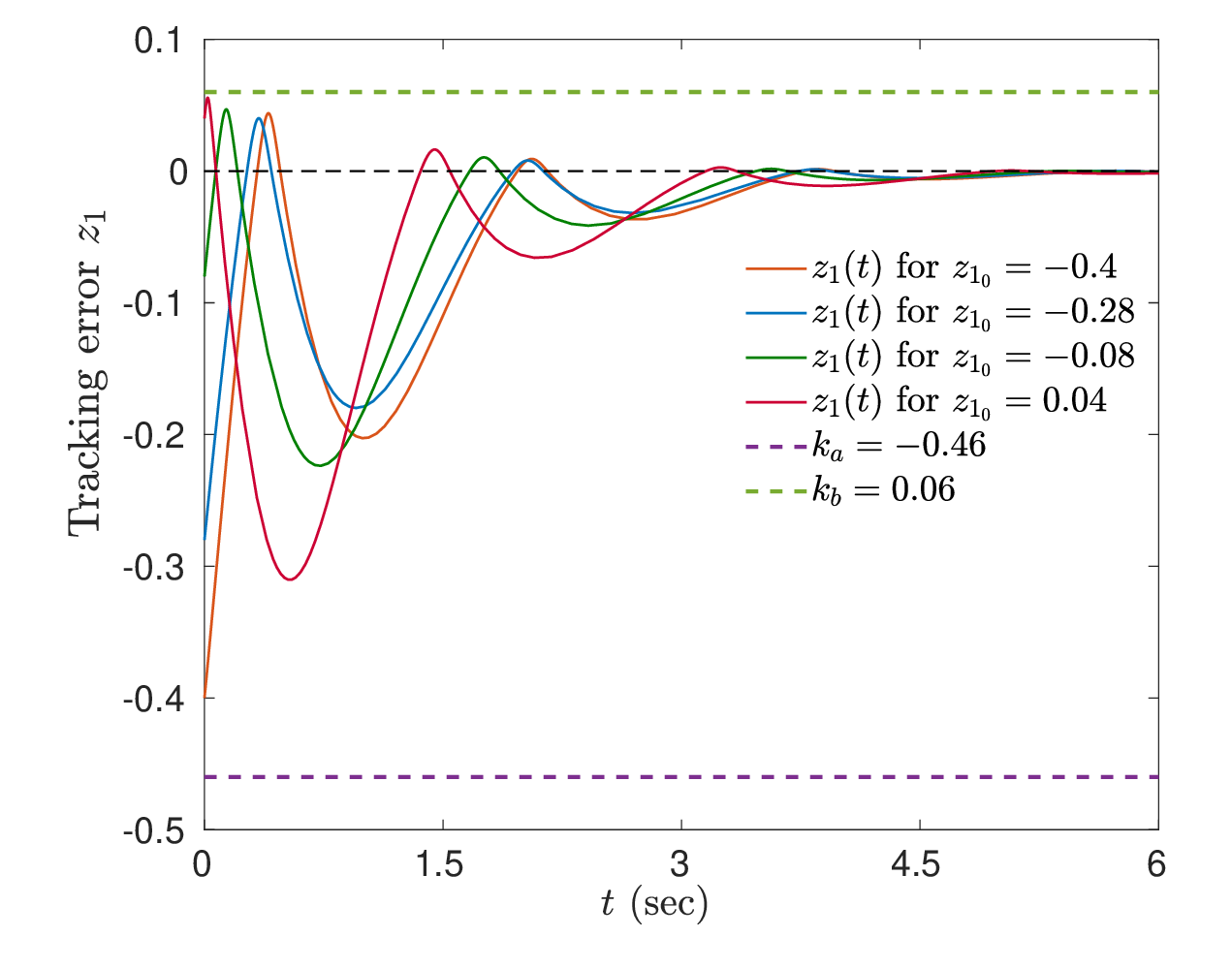}
        \caption{Tracking error $z_1(t)$ based on existing ABLF \cite{tee2009barrier}}
        \label{fig:z1_asymmetric_KPT}
    \end{subfigure}
    \caption{Tracking error $z_1(t)$ for various initial conditions when $k_a\ne k_b$}
    \label{fig:z1_asymmetric}
\end{figure*}
\begin{figure*}[h]
\centering
    \begin{subfigure}{0.47\textwidth}
        \centering
        \includegraphics[scale=0.41]{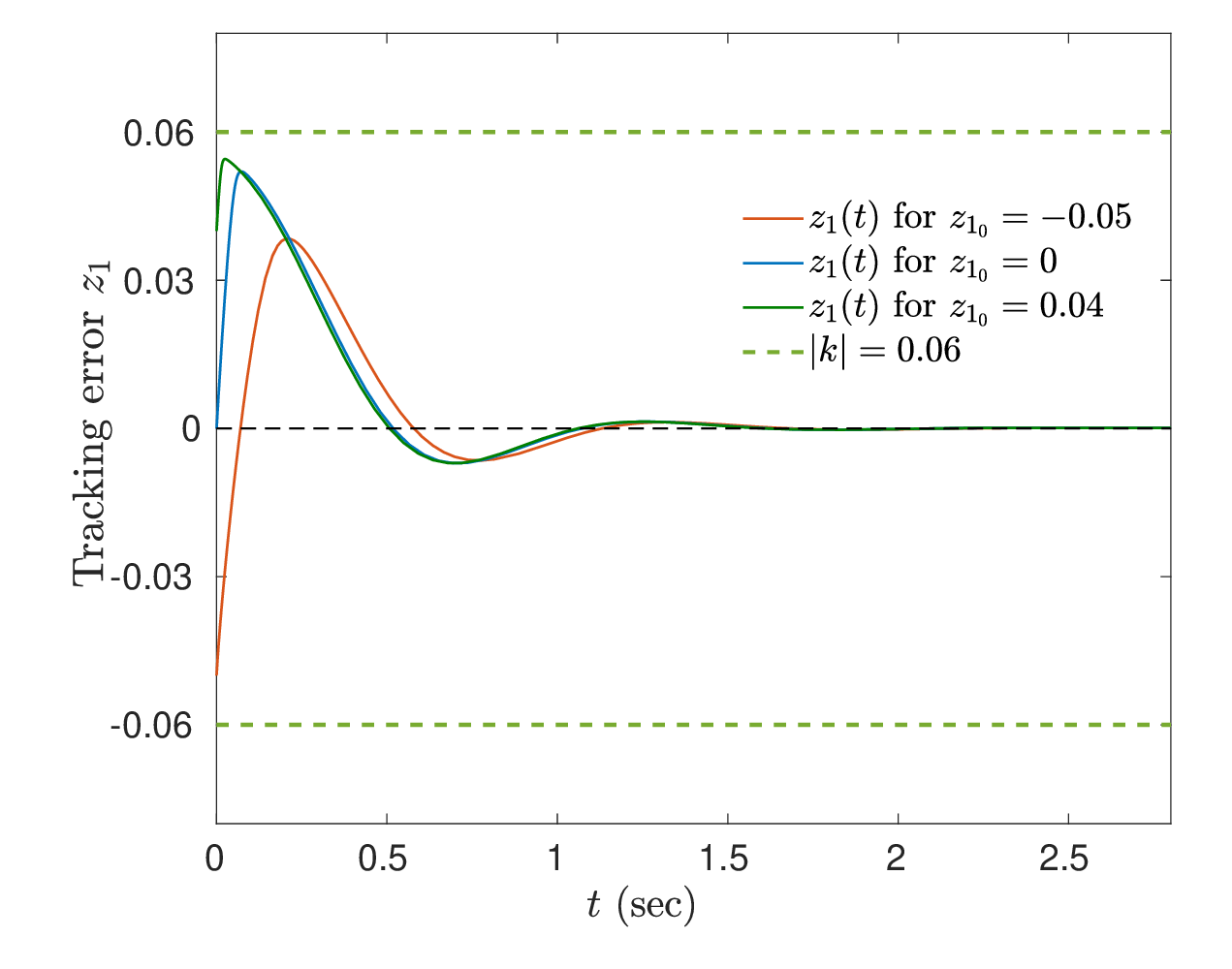}
        \caption{Tracking error $z_1(t)$ when proposed ABLF is used}
        \label{fig:z1_for_symmetric_proposed}
    \end{subfigure}
    \begin{subfigure}{0.47\textwidth}
        \centering
        \includegraphics[scale=0.43]{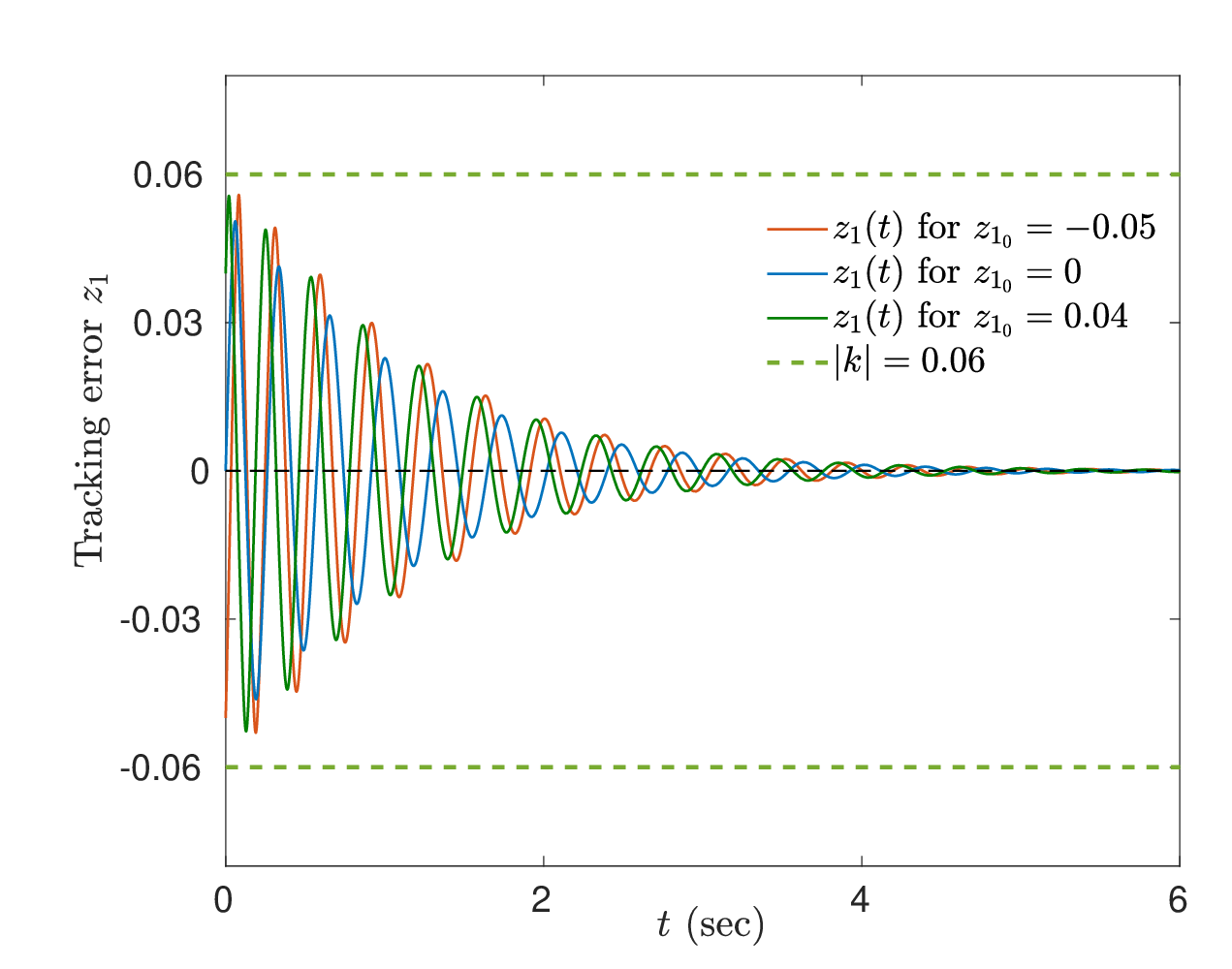}
        \caption{Tracking error $z_1(t)$ when existing SBLF from \cite{tee2009barrier} is used}
        \label{fig:z1_for_symmetric_case_KPT}
    \end{subfigure}
    \caption{Tracking error $z_1(t)$ for various initial conditions when $k_a= k_b$}
    \label{fig:z1_symmetric}
\end{figure*}

Next, we consider the case of symmetric constraints by choosing the bounds on $y_d(t)$ as $|y_{d}(t)|<0.5$. From this bound on $y_d$ and the bound on $x_1$, we obtain $k_a=k_b=0.06$. Fig.~ \ref{fig:z1_symmetric} shows $z_1(t)$ for initial conditions $x_1(0)\in \{0.15,0.2,0.24 \}$ (corresponding $z_{1_0}$'s are given by $z_{1_0}\in\{-0.05,0,0.04 \}$) with control gains $\kappa_1=\kappa_2=2$. From Fig.~ \ref{fig:z1_for_symmetric_proposed} and Fig.~ \ref{fig:z1_for_symmetric_case_KPT}, we see that performance under the proposed control is better when compared to the control scheme proposed in \cite{tee2009barrier}, based on the symmetric BLF. Note that the designed $u(t)$ is different owing to the construction for ensuring exponential stability. This holds despite the proposed ABLF reducing to the standard symmetric BLF from \cite{tee2009barrier} when $k_a=k_b$ and hence the difference in performance.

Let us now consider a third-order system given by
\begin{align}
    \begin{split}
        &\dot{x}_1=x^2_1(1-x_1)+x_2,~~\dot{x}_2=x_3\\
        & \dot{x}_3=u,~~y=x_1.
    \end{split}
\end{align}
The desired trajectory is given by $y_d=-0.2+0.3\sin{t}-0.5\cos{t}$, while the output is subject to the constraint $|x_1|<k_y=1$. The desired output trajectory satisfies $-0.2-\sqrt{0.34}<y_d(t)<-0.2+\sqrt{0.34}$ from which we select $Y_{lb}=0.79$ and $Y_{ub}=0.39$.  From \eqref{Eqn:ka_kb_def}, we now have $k_a=1-0.79=0.21$ and $k_b=1-0.39=0.61$. Fig.~\ref{fig:z1_order3sys} shows the tracking error $z_1(t)$ for initial conditions $x_1(0)\in \{-0.8,-0.65,-0.26,-0.15 \}$ (corresponding $z_{1_0}\in\{-0.1,0.05,0.44,0.55 \}$). To illustrate the difference in design for the third-order system, recall the ABLF in \cite{tee2009barrier}, given by \eqref{Eqn:ABLF_KPT}. Since the value of $p$ in $V_{\text{exist}}(z_1)$ must be an even integer satisfying the condition $p\geqslant n$, we choose $p=4$ as $n=3$. Therefore, we have
\small
\begin{align}
    V_{\text{exist}}=\frac{1}{4}\left[ q(z_1)~\text{log}\frac{k^4_b}{k^4_b-z^4_1} +(1-q(z_1))~\text{log}\frac{k^4_a}{k^4_a-z^4_1}\right],\notag
\end{align}
\normalsize
where $q=q(z_1)$ is a switching function and takes values as follows: $q(z_1)=1$ when $z_1>0$ and $q(z_1)=0$ otherwise. The equations for $\alpha_1(t),\alpha_2(t)$ and $u(t)$, designed using $V_{\text{exist}}$ based on the design steps in \cite{tee2009barrier}, are given by
\small
\begin{align}
    &\alpha_1=-x^2_1(1-x_1)+\dot{y}_d\notag\\
    &~~~~~~~-\kappa_1 \bigl[q(z_1)(k^4_b-z^4_1)+(1-q(z_1))(k^4_a-z^4_1)\bigr] z^3_1\notag\\
    &\alpha_2=\dot{\alpha}_1-\kappa_2 z_2-\left[\frac{q(z_1)}{k^4_b-z^4_1}+\frac{1-q(z_1)}{k^4_a-z^4_1} \right]z^{3}_1\notag\\
    &\alpha_3=u=\dot{\alpha}_2-\kappa_3 z_3-z_2.\notag
\end{align}
\normalsize
Due to the choice of $p$, we see that the control $u(t)$ obtained above has $z_1$ raised to higher powers, which may lead to very low control effort when $|z_1|$ is small. Therefore, it leads to slower convergence, which can be observed from Fig \ref{fig:z1_order3sys_KPT}. Since our proposed ABLF does not depend on the order of the system, the design is easily scalable to any higher order system without much effect on the performance of the control as seen from Fig \ref{fig:z1_order3sys_proposed}. Although we observe larger transient, as in the case of $z_{1_0}=0.55$ in Fig.~ \ref{fig:z1_order3sys_proposed}, this can be rectified by tuning the control gains $\kappa_1$, $\kappa_2$, and $\kappa_3$ as shown in Fig.~ \ref{fig:z1_various_gain_proposed}. For Fig.~ \ref{fig:z1_various_gain_proposed}, we consider the initial condition $x_1(0)=-0.15 ~(z_{1_0}=0.55)$ and plot the tracking error for various values of $\kappa_1,\kappa_2$ and $\kappa_3$. Observe that the overshoot and the speed of convergence improve while ensuring that the tracking error stays within the constraints. The increase in speed of convergence with the increase in gains can be explained from \eqref{Eqn:expo_bound_on_z1(t)}. Observe that $\tilde k=\min\{\kappa_1,2\kappa_2,2\kappa_3\}$, implying that an increase in the gains $\kappa_i,i=1,2,3$ increases $\tilde k$, and leads to faster convergence of the tracking error $z_1$.   
\begin{figure*}[h]
\centering
    \begin{subfigure}{0.47\textwidth}
        \centering
        \includegraphics[scale=0.43]{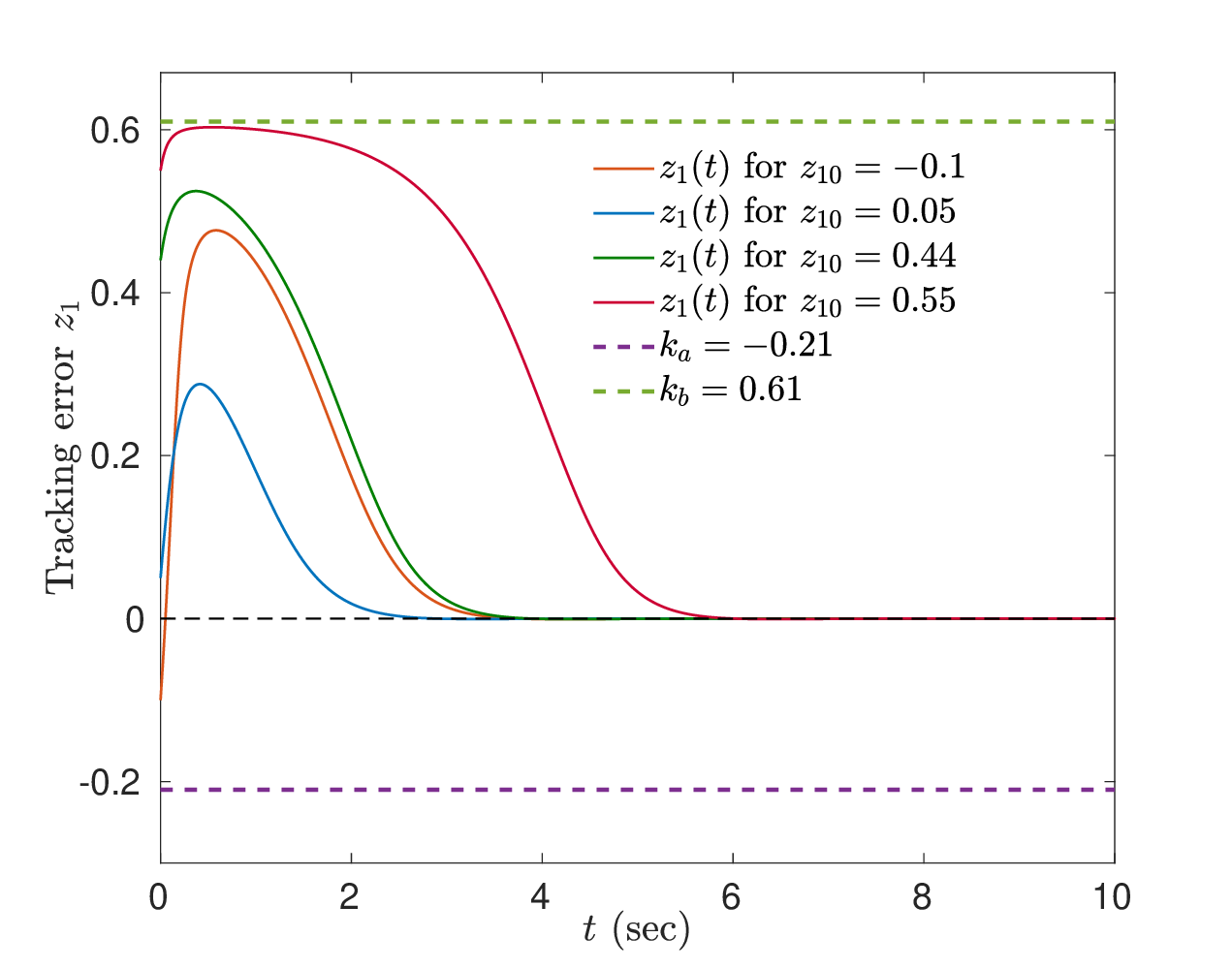}
        \caption{Tracking error $z_1(t)$ when proposed ABLF is used}
        \label{fig:z1_order3sys_proposed}
    \end{subfigure}
    \begin{subfigure}{0.47\textwidth}
        \centering
        \includegraphics[scale=0.44]{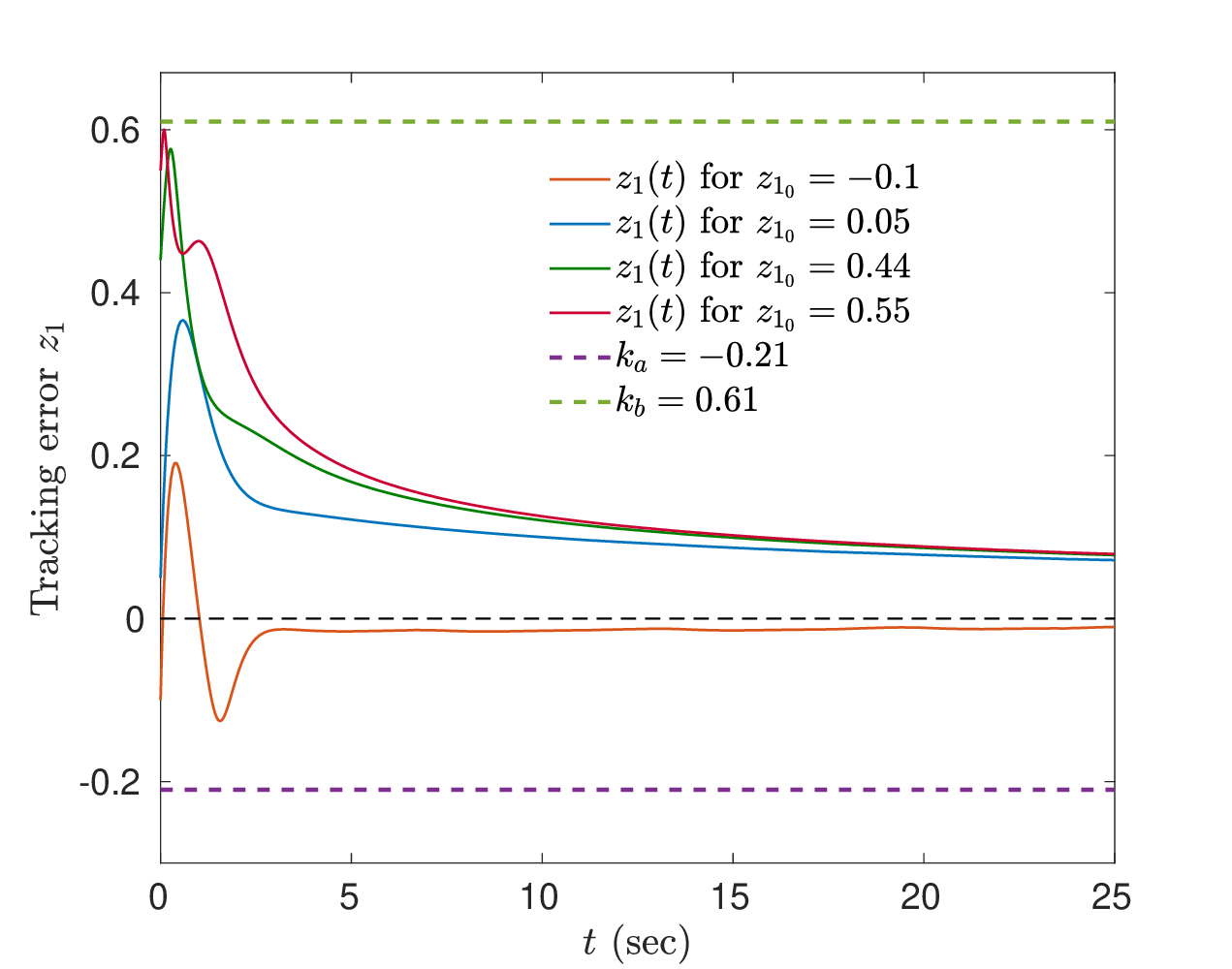}
        \caption{Tracking error $z_1(t)$ when existing ABLF from \cite{tee2009barrier} is used}
        \label{fig:z1_order3sys_KPT}
    \end{subfigure}
    \caption{Tracking error $z_1(t)$ for third order system for various initial conditions when $k_a< k_b$}
    \label{fig:z1_order3sys}
\end{figure*}
\begin{figure}[h!]
    \centering
    \includegraphics[scale=0.45]{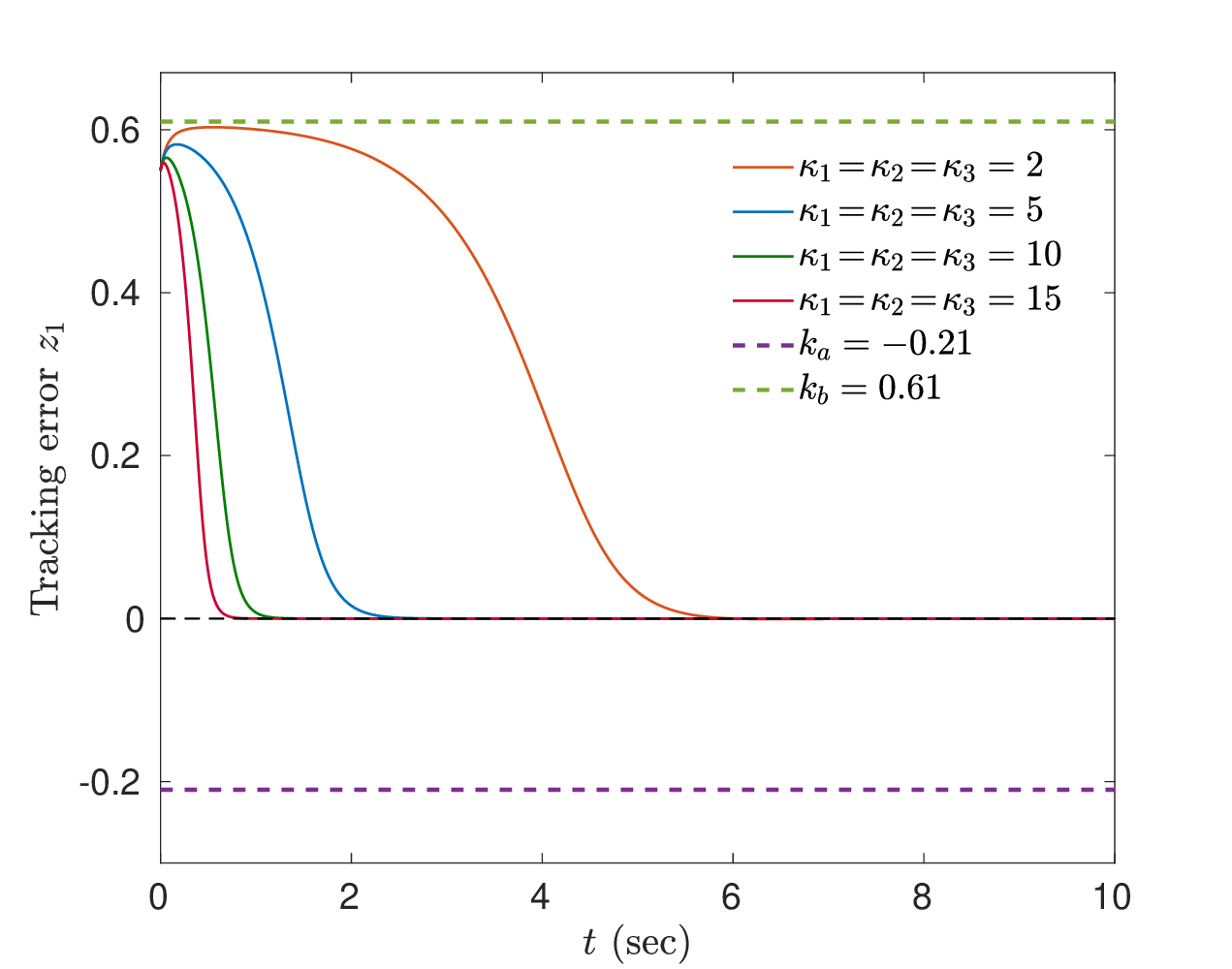}
    \caption{$z_1(t)$ for $z_{1_0}\!=\!0.55$ for various values of control gains}
    \label{fig:z1_various_gain_proposed}
\end{figure}

\section{Conclusion}
In this article, we presented a unified barrier Lyapunov function that addresses both symmetric and asymmetric output constraints. We showed that the proposed function is smooth and admits symmetric upper and lower bounding class $\mathcal{K}_{\infty}$ functions. Utilizing the unified barrier Lyapunov function, we proposed a control law for nonlinear single-input single-output systems with output constraints. The
proposed control law overcomes the disadvantages associated with the discontinuous switching function, thereby unifies the control design for both symmetric and asymmetric output constraints. Furthermore, we proved that under the proposed control law, the output tracking error converges to zero exponentially fast, while the output constraint requirements are satisfied at all times. We further illustrate the efficacy of the proposed control through numerical examples. 
%Furthermore, The proposed function generalizes the standard logarithmic SBLF.
% \appendices

% Appendixes, if needed, appear before the acknowledgment.
%\subsection{Footnotes}
\bibliographystyle{IEEEtran}
\bibliography{bibtex_database}

\end{document}